\documentclass[journal,twocolumn,a4paper]{IEEEtran}
\IEEEoverridecommandlockouts
\usepackage{cite}
\usepackage{amsmath,amssymb,amsfonts}
\usepackage{amsthm}
\usepackage{amsbsy}
\usepackage{bm}
\usepackage{booktabs}
\usepackage{tabularx}
\usepackage{multicol}
\usepackage{multirow}
\usepackage{caption}
\usepackage{subcaption}
\usepackage{graphicx}
\usepackage{textcomp}
\usepackage{color}
\usepackage{xcolor}
\usepackage{algorithm}
\usepackage{algpseudocode}
\algblock{Output}{EndOutput}
\algnotext{EndOutput}
\algnewcommand{\Initialize}[1]{%
	\State \textbf{Initialize:}
	\Statex \hspace*{\algorithmicindent}\parbox[t]{.8\linewidth}{\raggedright #1}
}
\usepackage[outdir=./]{epstopdf}
\usepackage{enumitem}
\newlist{steps}{enumerate}{1}
\setlist[steps, 1]{label = Step \arabic*:}
\allowdisplaybreaks
\def\BibTeX{{\rm B\kern-.05em{\sc i\kern-.025em b}\kern-.08em
\usepackage{hyperref}
    T\kern-.1667em\lower.7ex\hbox{E}\kern-.125emX}}

\newtheorem{proposition}{Proposition} 
\theoremstyle{remark}
\newtheorem{remark}{Remark}

\newtheorem{definition}{Definition}

\newcommand{\tUR}{UR}
\newcommand{\tUT}{UT}
\newcommand{\UR}{{\mathrm{r}}}
\newcommand{\UT}{{\mathrm{t}}}
\newcommand{\MN}{{\mathrm{c}}}
\newcommand{\nr}{{n}}
\newcommand{\nrl}{n'}
\newcommand{\NR}{{N_\mathrm{r}}}

\newcommand{\NT}{{N_\mathrm{t}}}
\newcommand{\nm}{p}
\newcommand{\NM}{N}
\newcommand{\UTUR}{\mathrm{tr}} 
\newcommand{\URUT}{\mathrm{rt}}

\newcommand{\yrtest}{\Tilde{\mathbf{y}}_{\nr}}
\newcommand{\ymtest}{\Tilde{\mathbf{y}}_{\nm}}
\newcommand{\phiu}{\bm{\varphi}_{\nr}}
\newcommand{\phiuh}{\bm{\varphi}_{\nr}^{H}}
\newcommand{\phid}{\bm{\Phi}_{\UT}}
\newcommand{\phidh}{\bm{\Phi}_{\UT}^H}
\newcommand{\esp}{\mathbb{E}}
\newcommand{\bmest}{\mathbf{b}_{\nm}}
\newcommand{\omegam}{\Tilde{\bm{\omega}}_{\nm}}
\newcommand{\Frj}{\mathbf{F}_\UR^\mathrm{J}}
\newcommand{\Fmj}{\mathbf{F}^{\mathrm{J}}}

\newcommand{\MSP}{\mathbf{MSP}}
\newcommand{\Wj}{\mathbf{W}_m^{\mathrm{J}}}

\algrenewcommand\algorithmicrequire{\textbf{Input:}}
\algrenewcommand\algorithmicensure{\textbf{Output:}}

\title{How to Proactively Monitor Untrusted Communications with Cell-Free Massive MIMO?}
\author{Isabella W. G. da Silva,~\IEEEmembership{Student Member,~IEEE,} Zahra~Mobini,~\IEEEmembership{Member,~IEEE}, Hien~Quoc~Ngo,~\IEEEmembership{Fellow,~IEEE}, Hyundong~Shin,~\IEEEmembership{Fellow,~IEEE},~and~Michail~Matthaiou,~\IEEEmembership{Fellow,~IEEE}
\thanks{This work is a contribution by Project REASON, a UK Government funded project under the Future Open Networks Research Challenge (FONRC) sponsored by the Department of Science Innovation and Technology (DSIT). It was also supported by the U.K. Engineering and Physical Sciences Research Council (EPSRC) (grants No. EP/X04047X/1 and EP/X040569/1). The work of Z.~Mobini and  H.~Q.~Ngo was supported by the U.K. Research and Innovation Future Leaders Fellowships under Grant MR/X010635/1. The work of M. Matthaiou was supported by the European Research Council (ERC) under the European Union’s Horizon 2020 research and innovation programme (grant agreement No. 101001331). The work of I. W. G. da Silva, H. Q. Ngo and M. Matthaiou was also supported by a research grant from the Department for the Economy
Northern Ireland under the US-Ireland R\&D Partnership Programme. The work of  H. Shin was supported by the National Research Foundation of Korea (NRF) grant funded by the Korean government (MSIT) under RS-2025-00556064 and by the MSIT (Ministry of Science and ICT), Korea, under the ITRC (Information Technology Research Center) support program (IITP-2025-2021-0-02046) supervised by the IITP (Institute for Information \& Communications Technology Planning \& Evaluation). (\textit{Corresponding authors: Hien Quoc Ngo; Hyundong Shin}.)
}
\thanks{I. W. G. da Silva,  H.~Q.~Ngo, and  M. Matthaiou are with the Centre for Wireless Innovation (CWI), Queen’s University Belfast, BT3 9DT Belfast, U.K. (e-mails: \{iwgdasilva01,  hien.ngo, m.matthaiou\}@qub.ac.uk). H.~Q.~Ngo is also with the Department of Electronic Engineering, Kyung Hee University, Yongin-si, Gyeonggi-do 17104, Republic of Korea.}
\thanks{Z.~Mobini is with the Centre for Wireless Innovation (CWI), Queen’s University Belfast, BT3 9DT Belfast, U.K., and also   with the Department of Electrical and Electronic Engineering, The University of Manchester, Manchester M13 9PL, U.K. (e-mail: zahra.mobini@qub.ac.uk).}

\thanks{H.~Shin is with the Department of Electronics and Information Convergence Engineering,
Kyung Hee University,
1732 Deogyeong-daero, Giheung-gu,
Yongin-si, Gyeonggi-do 17104, Republic of Korea
(e-mail: hshin@khu.ac.kr).}

\thanks{Parts of this paper were presented at the 2024 IEEE GLOBECOM~\cite{paperglobecom}.}}
\begin{document}

\maketitle

\begin{abstract}
This paper studies a cell-free massive multiple-input multiple-output (CF-mMIMO) proactive monitoring system in which multiple multi-antenna monitoring nodes (MNs) are assigned to either observe the transmissions from an untrusted transmitter (\tUT) or to jam the reception at the untrusted receiver (\tUR). We propose an effective channel state information (CSI) acquisition scheme for the monitoring system. In our approach, the MNs leverage the pilot signals transmitted during the uplink and downlink phases of the untrusted link and estimate the effective channels corresponding to the \tUT~and \tUR~via a minimum mean-squared error (MMSE) estimation scheme. We derive new spectral efficiency (SE) expressions for the untrusted link and the monitoring system. For the latter, the SE is derived for two CSI availability cases at the central processing unit (CPU); namely case-1: imperfect CSI knowledge at both MNs and CPU, case-2: imperfect CSI knowledge at the MNs and no CSI knowledge at the CPU. To improve the monitoring performance, we propose a novel joint mode assignment and jamming power control optimization method to maximize the monitoring success probability (MSP) based on the Bayesian optimization framework. Numerical results show that (a) our CF-mMIMO proactive monitoring system relying on the proposed CSI acquisition and optimization approach significantly outperforms the considered benchmarks; (b)   the MSP performance of our CF-mMIMO proactive monitoring system is greater than $0.8$, regardless of the number of antennas at the untrusted nodes or the precoding scheme for the untrusted transmission link.




\end{abstract}

\section{Introduction}\label{sec:Introduction}
Due to the ever-increasing number of devices and wireless capabilities and applications, data security against illegitimate users is a critical issue for the sixth-generation (6G) of wireless communications~\cite{10336902, 10684238}. To this end, physical-layer security (PLS) techniques, such as jamming~\cite{Zahra:TIFS:2019} and security-based precoding schemes~\cite{5580113}, are promising solutions to improve the secrecy of wireless communications from malicious attacks. Conventionally, PLS techniques are employed to preserve the confidentiality of the wireless communication assuming that a legitimate communication is being eavesdropped by an unauthorized device. However, with the advancement of infrastructure-free mobile communications, such as device-to-device (D2D) and mobile ad-hoc communications, new security challenges arise since malicious users may misuse these infrastructure-free networks to perform illegal activities, commit crimes, and jeopardize public safety. As a result, the need for novel wireless monitoring (a.k.a surveillance) approaches, aiming to eavesdrop on and/or intercept suspicious/untrusted messages through legitimate monitors, has gained sparkly attention recently~\cite{8014299,8726325}.

 In classical eavesdropping scenarios, the primary metric of security, known as \textit{secrecy capacity}, represents the maximum data rate at which information can be reliably received by a legitimate user without leakage of any valuable data to a potential eavesdropper~\cite{Zahra:TIFS:2019,6772207}. Analogously, in wireless monitoring systems, the channel conditions of the monitoring link should be better than that of the untrusted link to ensure that the suspicious message can be reliably decoded. This condition is hard to satisfy in several practical scenarios when a passive monitoring approach is considered, especially if the monitor is placed far away from the suspicious devices. Accordingly, to overcome this limitation, proactive monitoring has been proposed in~\cite{7321779} and~\cite{7880684}, aiming to improve the monitoring performance by cognitively sending jamming signals to degrade the channel conditions of the untrusted link while overhearing (eavesdropping) the suspicious message. 

 Several works have investigated proactive monitoring in a wide range of communication scenarios, as unmanned aerial vehicles (UAV) systems~\cite{ref:Zahra_proactive_UAV,10238835}, cognitive radio networks~\cite{9681707}, integrated sensing and communications (ISAC)~\cite{paperzonghan}, reconfigurable intelligent surface (RIS)-aided systems~\cite{9884994}, and fluid antenna systems (FAS)~\cite{10522668}, and for more complex scenarios with multiple untrusted links in~\cite{9673103,10015068,10384408}. For instance, in~\cite{10238835}, Guo \textit{et al.} considered a secure UAV-aided system where several legitimate UAVs collaboratively eavesdrop on the communication of suspicious links comprised of UAV transmitters and ground receivers. In~\cite{9681707}, the authors studied cooperative cognitive radios for proactive monitoring by considering that the secondary users help the primary users to eavesdrop on the untrusted link. In~\cite{paperzonghan}, proactive monitoring was exploited in an ISAC network, in which untrusted access points (APs) intend to illegally acquire the location of a sensed target.  
 RISs and FAS have also been investigated to enhance the monitoring performance for wireless monitoring systems~\cite{9884994,10522668}. Specifically, in~\cite{9884994},  an RIS was employed to improve the observing channel, whereas in~\cite{10522668}, a fluid antenna's legitimate monitor is considered. In addition, a number of works have investigated the scenario of proactive monitoring with multiple untrusted links. In~\cite{9673103}, Xu and Zhu optimized the average successful eavesdropping probability and the average eavesdropping rate of a scenario where numerous untrusted links are eavesdropped by one full-duplex (FD) monitor that is assumed to send jamming or constructive signals to the untrusted links. The authors also considered that the monitors were restricted to a quality-of-service (QoS) degradation constraint to avoid it being discovered by the untrusted links. Furthermore,~\cite{10015068} evaluated the relative eavesdropping rate considering that the suspicious links utilized PLS wiretap coding to protect their communication from being eavesdropped by the legitimate monitor. In~\cite{10384408}, the suspicious links were assumed to transmit jamming signals to defend against proactive monitoring. The authors formulated a Stackelberg game approach to track the interactions between the monitor and the suspicious user acting as a jammer, which was shown to enhance the successful eavesdropping probability at the monitor. Finally, in~\cite{9829333}, proactive monitoring via a spoofing relay approach was investigated for a multiple-input multiple-output (MIMO) orthogonal frequency division multiplexing (OFDM) system with directional beamforming.  
   \begin{table*}[t]
\centering
\caption{Our contributions in contrast to the state-of-the art}
\begin{tabular}{|c|c|c|c|c|c|c|}
\hline
Feature                       &\cite{9829333} &\cite{paperzahra}   &\cite{9381240}   &\cite{10032289}   &\cite{9709525}   & \textbf{our work} \\ \hline
CSI acquisition               &  &  & \checkmark &  &\checkmark  &   \checkmark                \\ \hline
Multiple MNs             &  &\checkmark &  &  &  &  \checkmark                 \\ \hline
Multi-antenna untrusted nodes &\checkmark & &  &\checkmark  &  &  \checkmark                 \\ \hline
Distributed MNs &  &\checkmark &  &  &  &  \checkmark                 \\ \hline
MSP maximization              & &   &   &   &   &    \checkmark               \\ \hline
\end{tabular}
\label{tab:contpaper}
\end{table*}

 Most of the aforementioned works rely on a single monitor, which should operate in FD mode to simultaneously observe the suspicious link and transmit jamming signals to interfere with the reception at the untrusted receiver (\tUR). To overcome this limitation, in a recent study~\cite{paperzahra}, the authors proposed a new proactive monitoring scheme, which exploits the cell-free massive MIMO (CF-mMIMO) infrastructure to enhance the monitoring capabilities within wireless surveillance frameworks. CF-mMIMO has been envisioned as one of the most promising technologies for 6G, as it overcomes the inherent intercell-interference  of traditional cellular systems and combines the concepts of distributed MIMO and massive MIMO, thereby availing of the benefits of both systems, such as high macro-diversity and ubiquitous coverage~\cite{7827017,10684260}. For wireless proactive monitoring systems, CF-mMIMO can enable a virtual FD mode by employing only half-duplex (HD) monitoring nodes (MNs). This way, the monitoring system becomes more cost-effective and less prone to self-interference. In particular, in~\cite{paperzahra}, the authors assumed two subsets of MNs: one subset for coherently observing the untrusted transmitters (\tUT s), and another subset to cooperatively jam the \tUR s. The observed signals at the MNs are forwarded to the central processing unit (CPU) for detecting the untrusted signals. It was observed that the proposed CF-mMIMO proactive monitoring system can significantly outperform the co-located mMIMO proactive monitoring system in terms of monitoring success probability (MSP).
 
Nevertheless, it is worth noting that a common assumption in the literature is the availability of global and perfect knowledge of channel state information (CSI) of the untrusted links at the MNs/CPU. However, in practical scenarios, only imperfect instantaneous CSI or statistical CSI is available at the MNs/CPU. In this context, Cheng \textit{et al.}~\cite{9381240} evaluated the impact of channel uncertainty on proactive monitoring. In this work, the authors formulated an optimization problem to enhance the monitoring performance under a covert constraint and showed that the uncertainty of the links can highly impact monitoring performance. In~\cite{10032289}, the authors studied the uplink of a multi-antenna proactive monitoring system with spatially correlated channels, considering imperfect instantaneous CSI for both the suspicious user at the multi-antenna suspicious receiver and the eavesdropping link at the monitor. Based on these considerations, they designed a transmit beamforming strategy for pilot and data jamming at the monitor to minimize the received signal-to-interference-plus-noise ratio (SINR) at the suspicious receiver. In~\cite{9709525}, Hu \textit{et al.} considered a simple setup of proactive monitoring, with single-antenna untrusted users and an FD dual-antenna legitimate monitor, and evaluated the effect of transmitting jamming signals to degrade the untrusted communication in the channel training phase of a suspicious communication. 


Motivated by the above, in this paper, we study a more advanced proactive monitoring system based on the CF-mMIMO technology, designed to monitor a pair of multi-antenna untrusted users through multiple multi-antenna MNs, under varying CSI knowledge assumptions at the CPU and MNs. More specifically, in the considered proactive monitoring system, the MNs are HD devices can operate in either observing or jamming mode. In particular, a group of MNs overhears the untrusted messages from the UT, while the remaining MNs send jamming signals to disrupt the UR. The monitoring performance of the CF-mMIMO proactive monitoring system is evaluated under two CSI knowledge cases;  namely \emph{1)} imperfect CSI knowledge at both the MNs and CPU \emph{2)} imperfect CSI knowledge at the MNs and no CSI knowledge at the CPU.  We would like to highlight that this work is an extension of our conference paper~\cite{paperglobecom}, where the CSI acquisition and SE expressions were first derived. In this work, we address the MSP maximization problem by jointly optimizing the MN mode assignment and transmit power control using Bayesian optimization, a method particularly effective for optimizing complex, hard-to-characterize objective functions. While Bayesian optimization has been primarily applied to hyperparameter tuning in machine learning algorithms~\cite{WU201926, pmlr-v54-klein17a}, its use in other optimization scenarios remains largely unexplored.  Moreover, in this work, we also provide a more detailed and complete performance evaluation of various system parameters on the MSP, and an asymptotic analysis as the number of MNs tends to infinity, which were not provided in~\cite{paperglobecom}. Table~\ref{tab:contpaper} benchmarks the contributions of our work in contrast to the state of the art. A more detailed description of the paper's are delineated below in a point-by-point format:
\begin{itemize}
    \item We propose a comprehensive transmission protocol for CF-mMIMO proactive monitoring systems, designed to monitor a pair of multi-antenna untrusted users using a simple yet effective CSI acquisition approach. In our CF-mMIMO proactive monitoring system,  we demonstrate that by leveraging pilot signals during both the uplink and beamforming training phases of the untrusted link, and by employing well-established and effective estimation techniques as the minimum mean-squared error (MMSE) technique, the MNs can estimate the effective channels to both the UT and UR and develop a  more practical monitoring system. Moreover, we evaluate the asymptotic behavior of the system as the number of MNs in the observing mode and jamming mode goes to infinity. For the former, we show that for a fixed number of MNs in jamming mode, as the number of MNs in observing mode increases, the effect of inter-MN interference and noise disappears at the CPU. For the latter, we show that we can downscale the necessary transmit power of each MN in jamming mode with the number of MNs in jamming mode, $M_\mathrm{J}$, by a factor of $\frac{1}{M_\mathrm{J}^2}$, while maintaining the SINR requirement for observing.  Unlike~\cite{paperzahra}, which assumed perfect CSI availability  for multiple untrusted single-antenna pairs, our analysis considers a more practical scenario involving a multi-antenna untrusted pair and realistic CSI acquisition via uplink and beamforming training. This leads to a fundamentally different system analysis and yields new insights into the asymptotic behavior, the impact of channel uncertainty, and interference management in CF-mMIMO surveillance systems.
    
    \item We derive analytical expressions for the spectral efficiencies (SEs) of the untrusted links and at the monitoring system, taking into account the proposed CSI acquisition and MMSE successive interference cancellation (MMSE-SIC) schemes. The MSP is then derived. Specifically, we investigate the monitoring performance of the CF-mMIMO proactive monitoring system under two CSI knowledge cases;  namely case-1: imperfect CSI knowledge at both the MNs and CPU, and case-2: imperfect CSI knowledge at the MNs and no CSI knowledge at the CPU. 
    It is worth emphasizing that, while the expressions derived in this work are tailored for a CF-mMIMO system, their versatility extends to various monitoring scenarios. 
    \item Numerical results validate the following:
(a)  The effectiveness of our CF-mMIMO proactive monitoring system, which relies on the proposed CSI acquisition and MSP optimization approach, that can significantly outperform benchmarks, such as random mode assignment and equal jamming power control; (b) The CF-mMIMO proactive monitoring system achieves higher MSP under case-1 compared to case-2 across almost all evaluated scenarios, particularly when the number of antennas at the untrusted nodes is high or when   maximum-ratio transmission (MRT) precoding scheme is applied at the UT. Moreover, with  zero-forcing (ZF) precoding, the MSP performance loss of case-2 relative to case-1 is reduced across various scenarios;  (c) The proposed CF-mMIMO proactive monitoring system  provides remarkable MSP gains compared to a co-located mMIMO-aided proactive monitoring system.
\end{itemize}

\textit{Notation:} Throughout this paper, bold upper-case letters denote matrices, whereas bold lower-case letters denote vectors; $(\cdot)^T$ and $(\cdot)^H$ stand for the matrix transpose and Hermitian transpose, respectively; $\mathbf{I}_M$ is the $M \times M$ identity matrix; $||\cdot||$ and $|\cdot|$ are the Euclidean-norm and the absolute value operator; $\det(\cdot)$  and $\esp\{\cdot\}$ is the  determinant and the expectation operator, respectively, while $\operatorname{Var}(a)$$\triangleq$$\esp\left\{|a - \esp\{a\}|^2\right\}$ is the variance operator. A circular symmetric complex Gaussian vector $\mathbf{z}$ with covariance matrix $\mathbf{C}$ is denoted by $\mathbf{z}\sim\mathcal{CN}(\mathbf{0}, \mathbf{C})$. For convenience, the main symbols used in this paper are presented in Table~\ref{tab:symbols}.

\begin{table}
\centering
\caption{List of notations}
\label{tab:symbols}
\begin{tabular}{l c p{0.75\linewidth}}
\toprule
\multicolumn{1}{l|}{$M$}              && Number of MNs\\
\multicolumn{1}{l|}{$\NR$}              & & Number of antennas at the \tUR\\
\multicolumn{1}{l|}{$\NT$}              & & Number of antennas at the \tUT\\
\multicolumn{1}{l|}{$\NM$}              && Number of antennas at the MNs\\
\multicolumn{1}{l|}{$\alpha_m$}                 && MN operation assignment\\
\multicolumn{1}{l|}{$\bm{\varphi}_n$}                 && Pilot  vector transmitted by the $n$th antenna of the \tUR\\
\multicolumn{1}{l|}{$\bm{\Phi}_\UT$}                 && Pilot  matrix transmitted by the \tUT\\
\multicolumn{1}{l|}{$\tau_\UR$}              && Uplink pilot length\\
\multicolumn{1}{l|}{$\tau_\UT$}                 && Downlink pilot length\\
\multicolumn{1}{l|}{$\rho_\UR$}        && Maximum normalized power at the \tUR\\
\multicolumn{1}{l|}{$\rho_\UT$}        && Maximum normalized power at the \tUT\\
\multicolumn{1}{l|}{$\mathbf{Y}_{\UTUR}$}              && Received pilot at the \tUT~from the \tUR\\
\multicolumn{1}{l|}{$\mathbf{Y}_{m\UR}$}              && Received pilot at the $m$th MN from the \tUR\\
\multicolumn{1}{l|}{$\mathbf{y}_{\UR}$}              && Received signal at the \tUR\\
\multicolumn{1}{l|}{$\mathbf{y}_{m}$}              && Received signal at the $m$th MN\\
\multicolumn{1}{l|}{$\mathbf{G}_{\UTUR}$}            && Channel matrix between the \tUR~and the \tUT\\
\multicolumn{1}{l|}{$\mathbf{G}_{m\UR}$}              && Channel matrix between the \tUR~and the $m$th MN\\
\multicolumn{1}{l|}{$\mathbf{G}_{\UT m}$}            && Channel matrix between the \tUT~and the $m$th MN\\
\multicolumn{1}{l|}{$\mathbf{G}_{mm'}$}              && Channel matrix between MN $m$ and MN $m'$\\
\multicolumn{1}{l|}{$\beta_{\UTUR}$}              && Large-scale fading coefficient between the \tUR~and the \tUT\\
\multicolumn{1}{l|}{$\beta_{m\UR}$}             && Large-scale fading coefficient between \tUR~and MN $m$\\
\multicolumn{1}{l|}{$\beta_{mm'}$}             && Large-scale fading coefficient between MNs $m$ and $m'$\\
\multicolumn{1}{l|}{$\mathbf{H}_{\UTUR}$}             && Small-scale fading between the \tUR~and the \tUT\\
\multicolumn{1}{l|}{$\mathbf{H}_{m\UR}$}               && Small-scale fading between the \tUR~and the $m$th MN\\
\multicolumn{1}{l|}{$\mathbf{W}$}             && Precoding matrix \\
\multicolumn{1}{l|}{$\mathbf{Y}_{\URUT}$}              && Received pilot matrix at the \tUR~from the \tUT\ \\
\multicolumn{1}{l|}{$\mathbf{Y}_{m\UT}$}                && Received pilot matrix at the $m$th MN from the \tUT\\
\multicolumn{1}{l|}{$\mathbf{x}_{\URUT}$}                 && Symbol vector intended to the \tUR\\
\multicolumn{1}{l|}{$\mathbf{x}^\mathrm{J}$}                 && Jamming symbol intended to the \tUR\\
\multicolumn{1}{l|}{$\mathbf{s}_{\UT}$}                && Transmitted signal from the \tUT\\
\multicolumn{1}{l|}{$\mathbf{s}_{m}^\mathrm{J}$}                  && Transmitted signal by the $m$th MN in jamming mode\\
\multicolumn{1}{l|}{${\mathbf{z}}_{\MN}$}                  && Aggregated received signal at the CPU\\
\multicolumn{1}{l|}{${\mathbf{V}}_m$}                  && MMSE combining matrix \\
\multicolumn{1}{l|}{$\rho_\mathrm{J}$}                && Transmit power at the MNs in jamming mode\\
\multicolumn{1}{l|}{$\mathrm{SE}_\UR$}                && Achievable SE at the \tUR\\
\multicolumn{1}{l|}{$\mathrm{SE}_\MN$}                && Achievable SE at the CPU\\
\multicolumn{1}{l|}{$\bm{\Theta}_{\UR}$}                && Side information available at the \tUR\\
\multicolumn{1}{l|}{$\bm{\Theta}_{\MN}$}                && Side information available at the CPU\\
\bottomrule
\end{tabular}  
\end{table}

 The remainder of this paper is organized as follows: {Section II describes the  CF-mMIMO proactive monitoring system model. Also in this section, the uplink and downlink channel estimation phases are presented.} In Section III, the SE expressions for the unstrusted communication link and proactive monitoring system are derived. In Section IV, we present an efficient algorithm for joint optimization of the MN mode assignment and jamming power allocation to maximize the MSP. Section V presents the numerical results and discussions, while Section VI concludes the paper.




\section{System Model}\label{sec:systemmodel}
\begin{figure}[t]
    \includegraphics[width=0.45\textwidth]{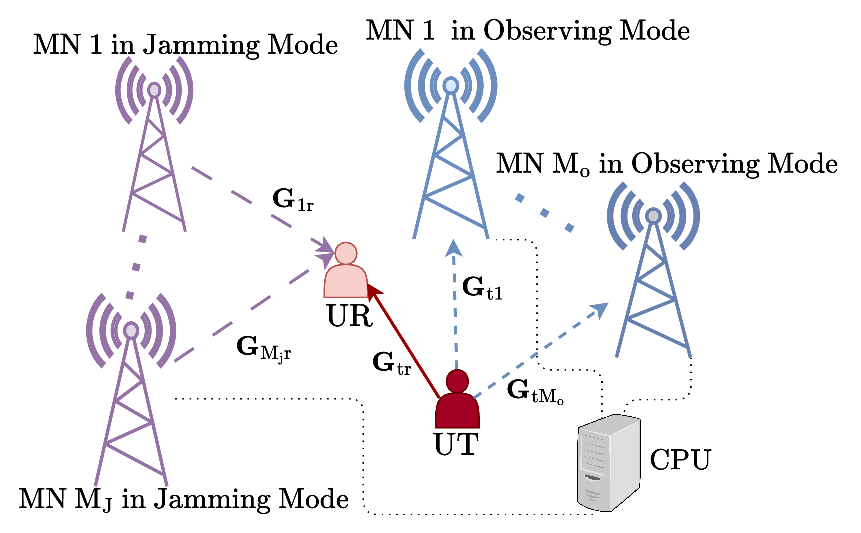}
    \label{fig:cellsize}
   \caption{System model of a CF-mMIMO proactive monitoring system with MNs operating in observing mode or jamming mode.}
\label{fig:systemmodel}
\end{figure}
 
As illustrated in Fig.~\ref{fig:systemmodel}, a CF-mMIMO proactive monitoring system is considered. This network consists of $M$ MNs with $\NM$ antennas, and a communication untrusted pair, in which the \tUT~and \tUR\footnote{Throughout this work, we use the subscripts $\UT$, $\UR$, and  $\MN$ to refer to the \tUT, \tUR, and  CPU, respectively.} are equipped with $\NT$ and $\NR$ antennas, respectively. All nodes are assumed to operate in HD mode. 
To monitor the untrusted pair during data transmission, a number of the MNs is assigned to operate in observing mode, where they receive untrusted messages from the \tUT, while the rest of the MNs operate in jamming mode, where they send jamming signals to disrupt the reception at the \tUR.

\begin{figure*}[!tbp]
\centering
  \includegraphics[width=0.65\linewidth]{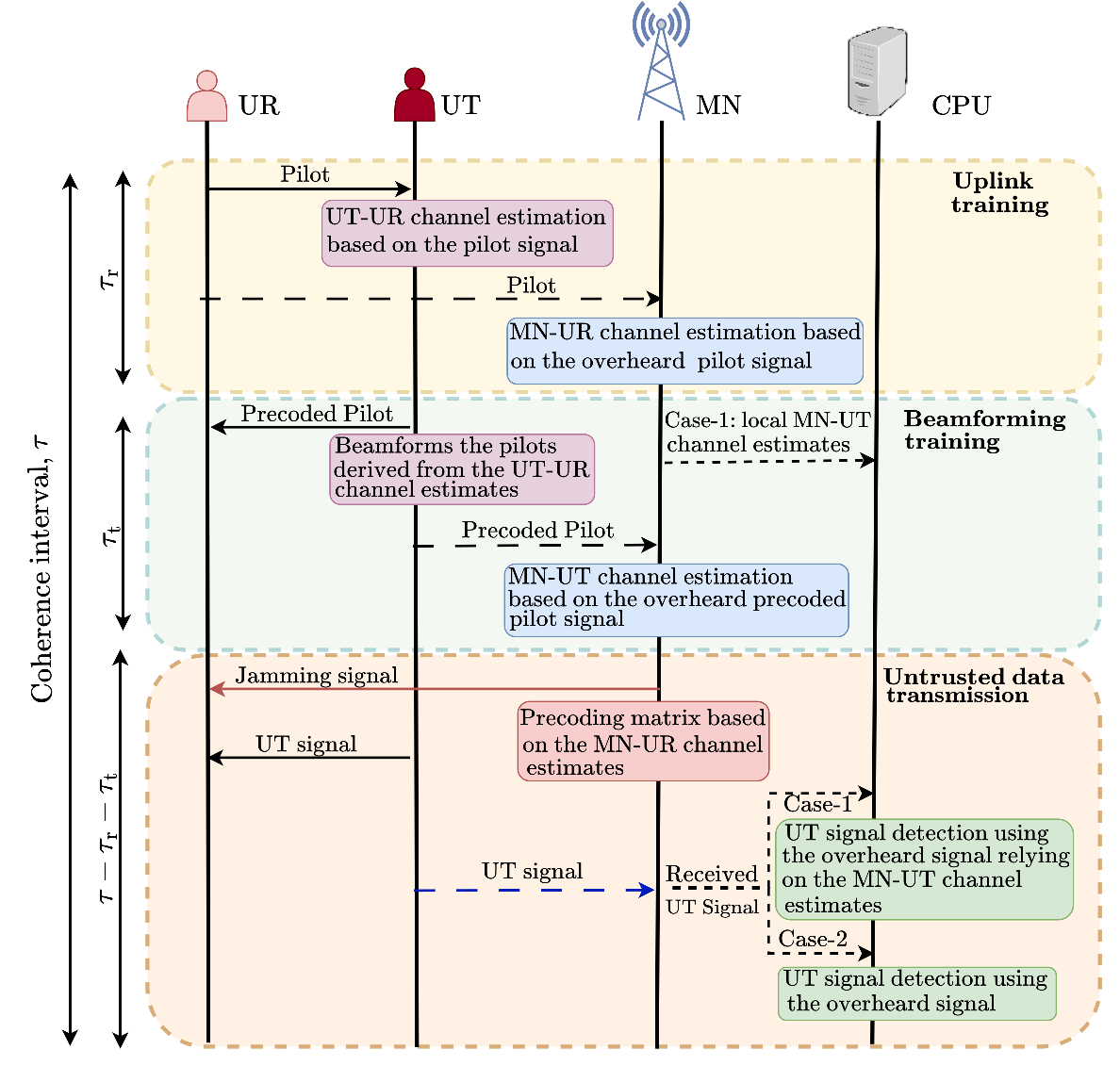}
    \caption{Flow chart of the three-stage transmission procedure for the CF-mMIMO monitoring system.}
    \label{fig:enter-label}
\label{fig:flowchart}
\hrulefill
\end{figure*}

A practical training-based transmission divided into uplink training and beamforming training is considered for the untrusted link. During the uplink training, the \tUR~transmits uplink pilots to the \tUT, enabling the \tUT~to estimate the channel to the \tUR. These channel estimates are then employed by the \tUT~in the beamforming training phase to define the precoder for transmitting the information signal to the \tUR. Also, during this phase, the \tUR~acquires knowledge of the channel gains via beamforming schemes to accurately detect the downlink signals transmitted from the \tUT~\cite{6736537}. Simultaneously, the MNs can leverage the pilot signals transmitted, during both the uplink training and downlink beamforming training phases of the untrusted link, to estimate the channels to both the \tUR~and \tUT.  This can be accomplished because the \tUT~and \tUR~are assumed to be part of the network infrastructure, albeit untrusted.\footnote{ Similar to~\cite{7321779} and~\cite{7880684}, we assume that the untrusted pair have been prior detected by authorized parties. Please refer to~\cite{8014299} for details on how to detect the untrusted parties.} Therefore, it is also feasible to consider that they are operating with known network protocols, and thus, the transmitted pilot sequences are known by the untrusted pair and by the MNs~\cite{8360138}. Accordingly,  the details of the uplink training, beamforming training, and downlink data transmission phases are provided in the next subsections and illustrated in a flow chart in Fig.~\ref{fig:flowchart}.
\vspace{-0.3em}
\subsection{Uplink Training}\label{sec:upt}
Let us assume that the pilot sequence transmitted from the $\nr$th antenna of the \tUR~to the \tUT~is denoted by $\phiu\in\mathbb{C}^{\tau_{\UR}\times1}$,  where $\tau_\UR$ represents the pilot sequence length. All pilot sequences are considered pairwisely orthonormal, i.e., $\phiuh\bm{\varphi}_{\nrl}=0$ if $\nr\neq\nrl$, with $||\phiu||^2=1$, $\forall\nr=1,\dots,\NR$. Thus, it is required that $\tau_{\UR}\geq\NR$. Simultaneously, the MNs also receive the pilot signals transmitted from the \tUR. Therefore, the received $\NT\times\tau_{\UR}$ pilot matrix at the \tUT, and the $\NM\times\tau_{\UR}$ pilot matrix  at the $m$th MN are, respectively, given by~\cite{9392380}
\vspace{-0.5em}
\begin{align}
   \mathbf{Y}_{\UTUR}&=\sqrt{\tau_{\UR} \rho_\UR} \sum_{\nr=1}^{\NR} \mathbf{g}_{\UTUR, \nr} \phiuh+\bm{\Omega}_{\UT},\\
    \mathbf{Y}_{m\UR}&=\sqrt{\tau_{\UR} \rho_\UR} \sum_{\nr=1}^{\NR} \mathbf{g}_{m\UR, \nr} \phiuh+\bm{\Omega}_{m\UR},
\end{align}
where $\rho_\UR$ is the normalized power of each pilot symbol
transmitted by the \tUR, while $\bm{\Omega}_{\UT}\in\mathbb{C}^{\NT\times\tau_{\UR}}$ and $\bm{\Omega}_{m\UR}\in\mathbb{C}^{\NM\times\tau_{\UR}}$ are the noise matrices at the \tUT~and at the $m$th MN, which are assumed to have independent and identically distributed (i.i.d.) $\mathcal{CN}(0,1)$ entries.  Also, $\mathbf{g}_{\UTUR, \nr}$ is the $\nr$th column of the channel matrix from the \tUR~to the \tUT, denoted by $\mathbf{G}_{\UTUR}$, while $\mathbf{g}_{m\UR, \nr}$ is the $\nr$th column of the channel matrix from the \tUR~to the $m$th MN,  denoted by $\mathbf{G}_{m\UR}$.  Accordingly, 
  $\mathbf{G}_{\UTUR}\in\mathbb{C}^{\NT\times \NR}$ and $\mathbf{G}_{m\UR}\in\mathbb{C}^{\NM\times \NR}$ are modeled as
\begin{align}\label{eq:channeldef}
\mathbf{G}_{\UTUR}&=\sqrt{\beta_{\UTUR}}\mathbf{H}_{\UTUR},\\
\mathbf{G}_{m\UR}&=\sqrt{\beta_{m\UR}}\mathbf{H}_{m\UR},
\end{align}
where $\beta_{\UTUR}$ and $\beta_{m\UR}$ are large-scale fading coefficients, while $\mathbf{H}_{\UTUR}$ and $\mathbf{H}_{m\UR}$ are the associated small-scale fading matrices between the \tUR~and \tUT~and between the \tUR~and the $m$th MN, respectively. The entries of the small-scale fading matrices are assumed to be i.i.d. $\mathcal{CN}(0,1)$. The pilot signals received at the \tUT~and at the $m$th MN are projected onto $\phiu$, allowing the channel vector for each antenna of the \tUR~to be estimated by the \tUT~and MN $m$, as
\begin{align}
\Tilde{\mathbf{y}}_{\UTUR,\nr} & \triangleq\mathbf{Y}_{\UTUR} \phiu
=\sqrt{\tau_{\UR} \rho_{\UR}} \mathbf{g}_{\UTUR, \nr}+\bm{\Omega}_{\UT}\phiu, \label{eq:yRT}\\
\Tilde{\mathbf{y}}_{m\UR,\nr} & \triangleq\mathbf{Y}_{m\UR} \phiu
=\sqrt{\tau_{\UR} \rho_{\UR}} \mathbf{g}_{m\UR, \nr}+\bm{\Omega}_{m\UR}\phiu, \label{eq:yRm}
\end{align}
respectively. Based on \eqref{eq:yRT} and \eqref{eq:yRm}, a linear MMSE approach can be employed to attain the estimates of the channel responses $\mathbf{g}_{\UTUR, \nr}$ and $\mathbf{g}_{m\UR, \nr}$ as
\begin{align}\label{eq:hatg}
    \hat{\mathbf{g}}_{\UTUR, \nr}&=\esp\left\{\mathbf{g}_{\UTUR, \nr}\Tilde{\mathbf{y}}_{\UTUR,\nr}^H\right\}\!\left(\esp\left\{\Tilde{\mathbf{y}}_{\UTUR,\nr}\Tilde{\mathbf{y}}_{\UTUR,\nr}^H\right\}\right)^{-1}\Tilde{\mathbf{y}}_{\UTUR,\nr}\nonumber\\
    &=\frac{\sqrt{\tau_{\UR}\rho_{\UR}}\beta_{\UTUR}}{\tau_{\UR}\rho_{\UR}\beta_{\UTUR}+1}\Tilde{\mathbf{y}}_{\UTUR,\nr},\\
    \hat{\mathbf{g}}_{m\UR, \nr}\!&=\esp\left\{\mathbf{g}_{m\UR, \nr}\Tilde{\mathbf{y}}_{m\UR,\nr}^H\right\}\!\left(\esp\left\{\Tilde{\mathbf{y}}_{m\UR,\nr}\Tilde{\mathbf{y}}_{m\UR,\nr}^H\right\}\right)^{-1}\Tilde{\mathbf{y}}_{m\UR,\nr}\nonumber\\
    &=\frac{\sqrt{\tau_{\UR}\rho_{\UR}}\beta_{m\UR}}{\tau_{\UR}\rho_{\UR}\beta_{m\UR}+1}\Tilde{\mathbf{y}}_{m\UR,\nr},
\end{align}
respectively. From the property of the MMSE estimator, the components of $\hat{\mathbf{g}}_{\UTUR, \nr}$ and $\hat{\mathbf{g}}_{m\UR, \nr}$ are i.i.d. Gaussian, with a mean-square given by
\begin{align}
    \gamma_{\UTUR}&\triangleq\esp\{||\hat{\mathbf{g}}_{\UTUR, \nr}||^2\}=\frac{\tau_{\UR}\rho_{\UR}\beta_{\UTUR}^2}{\tau_{\UR}\rho_{\UR}\beta_{\UTUR}+1},\\
    \gamma_{m\UR}&\triangleq\esp\{||\hat{\mathbf{g}}_{m\UR, \nr}||^2\}=\frac{\tau_{\UR}\rho_{\UR}\beta_{m\UR}^2}{\tau_{\UR}\rho_{\UR}\beta_{m\UR}+1},
\end{align}
respectively.
\vspace{-0.5em}
\subsection{Beamforming Training}\label{sec:beamtrain}
In the beamforming training phase, the \tUT~beamforms the pilots using a precoding matrix derived from the channel estimate of the \tUR~obtained during the uplink training phase.  Note that this precoding matrix is deployed at the UT for transmitting data to the UR. Let us denote the precoding matrix by $\mathbf{W}=\Big[\frac{\mathbf{w}_{1}}{||\mathbf{w}_1||},\dots,\frac{\mathbf{w}_{\NR}}{||\mathbf{w}_\NR||}\Big]$, $\mathbf{W} \in\mathbb{C}^{\NT\times\NR}$, with $\frac{\mathbf{w}_{\nr}}{||\mathbf{w}_\nr||}$, $\nr=1,\dots,\NR$, being the intended normalized $N_{\UT}\times1$ precoding vector for each antenna of the \tUR. 
 Let $\phid\in\mathbb{C}^{\NR\times \tau_{\UT}}$ be the pilot sequence matrix from the \tUT~to \tUR, with $\tau_{\UT}$ being the duration (in symbols) of the beamforming training.  We assume that the rows of $\phid$ are pairwisely orthogonal, i.e., $\phid\phidh=\mathbf{I}_\NR$. Hence, it is required that $\tau_{\UT}\geq\NR$. Then, the received pilot matrix at the \tUR~$\in \mathbb{C}^{\NR \times \tau_\UT}$, and at the $m$th MN $\in \mathbb{C}^{\NM \times \tau_\UT}$ are given by 
\begin{align}
    \mathbf{Y}_{\URUT} &= \sqrt{\tau_{\UT}\rho_{\UT}}\mathbf{G}_{\UTUR}^H\mathbf{W}\phid + \bm{\Omega}_{\UR},\\
    \mathbf{Y}_{m\UT} &= \sqrt{\tau_{\UT}\rho_{\UT}}\mathbf{G}_{\UT m}^H\mathbf{W}\phid + \bm{\Omega}_{m\UT},
\end{align}
respectively, where $\rho_\UT$ is the normalized power of each pilot symbol
transmitted by \tUT, whereas $\mathbf{G}_{\UT m}\in\mathbb{C}^{\NT\times \NM}$ is the channel response between MN $m$ and \tUT, modeled as
\begin{align}
    \mathbf{G}_{\UT m}&=\sqrt{\beta_{\UT m}}\mathbf{H}_{\UT m},\label{eq:gtm}
\end{align}
where $\beta_{\UT m}$ is the large-scale fading coefficient and $\mathbf{H}_{\UT m}$ is the small-scale fading matrix between the $m$th MN and \tUT, with i.i.d. $\mathcal{CN}(0,1)$ entries. In addition, $\bm{\Omega}_{\UR}\in\mathbb{C}^{\NR\times\tau_{\UT}}$ and $\bm{\Omega}_{m\UT}\in\mathbb{C}^{\NM\times\tau_{\UT}}$ are the noise matrices at the \tUR~and at the $m$th MN, respectively, which are assumed to have i.i.d. $\mathcal{CN}(0,1)$ elements. As discussed in~\cite{6736537}, we can project $\phid$ onto $\mathbf{Y}_{\URUT}$ and $\mathbf{Y}_{m\UT}$, and use it to estimate the effective channels. Accordingly, 
\begin{align}
    \Tilde{\mathbf{Y}}_{\URUT}&\triangleq\mathbf{Y}_{\URUT}\phidh=\sqrt{\tau_{\UT}\rho_{\UT}}\mathbf{G}_{\UTUR}^H\mathbf{W}+\Tilde{\bm{\Omega}}_{\UR},\label{eq:yrtp}\\
    \Tilde{\mathbf{Y}}_{m\UT}&\triangleq\mathbf{Y}_{m\UT}\phidh=\sqrt{\tau_{\UT}\rho_{\UT}}\mathbf{G}_{\UT m}^H\mathbf{W}+\Tilde{\bm{\Omega}}_{m\UT},\label{eq:ymtp}
\end{align}
 where $\Tilde{\bm{\Omega}}_{\UR}\triangleq\bm{\Omega}_{\UR}\phidh$ and $\Tilde{\bm{\Omega}}_{m\UT}\triangleq\bm{\Omega}_{m\UT}\phidh$. Let us define $\mathbf{A}_\UR\triangleq\mathbf{G}_{\UTUR}^H\mathbf{W}$, with entries given by $a_{\nr,\nrl}\triangleq\mathbf{g}_{\UTUR,\nr}^H \mathbf{w}_{\nrl}$, and $\mathbf{B}_m=[\mathbf{b}_{1},\dots,\mathbf{b}_{\NM}]^H$ with $\bmest^H\triangleq\mathbf{g}_{\UT m,\nm}^H\mathbf{W}$. From \eqref{eq:yrtp} and \eqref{eq:ymtp}, the received pilot vector at each antenna of the \tUR~and MN $m$ are given, respectively, by
\begin{align}
    {\yrtest} &= \sqrt{\tau_{\UT}\rho_{\UT}}\mathbf{a}^T+ \Tilde{\bm{\omega}}_{\nr},\\
    {\ymtest}&= \sqrt{\tau_{\UT}\rho_{\UT}}\bmest^H+ \omegam,\label{eq:ynmd}
\end{align}
where $\mathbf{a}\triangleq[a_{\nr,1},\dots,a_{\nr,\NR}]$ is the $\nr$th column of $\mathbf{A}_\UR$, while ${\yrtest}$ and ${\ymtest}$ are the $n$th and $p$th columns of $\Tilde{\mathbf{Y}}_{\URUT}$ and $\Tilde{\mathbf{Y}}_{m\UT}$, respectively. Also, $\Tilde{\bm{\omega}}_{\nr}$ and $\omegam$ is the $n$th and $p$th column of $\Tilde{\bm{\Omega}}_{\UR}$ and $\Tilde{\bm{\Omega}}_{m\UT}$, respectively. 

\begin{proposition}
Assuming that $a_{\nr,1},\dots,a_{\nr,\NR}$ and $\bmest$ can be estimated independently, based on $\yrtest$ and $\ymtest$, the MMSE channel estimate of $a_{\nr,\nrl}$ and of $\bmest$ are written as
\begin{align}
    \hat{a}_{\nr,\nrl}
    =&~\esp\left\{a_{\nr,\nrl}\right\}+\frac{\sqrt{\tau_{\UT} \rho_{\UT}} \operatorname{Var}(a_{\nr,\nrl})}{\tau_{\UT} \rho_{\UT} \operatorname{Var}(a_{\nr,\nrl})+1}\nonumber\\
    &\times(\tilde{y}_{\nrl}-\sqrt{\tau_{\UT} \rho_{\UT}} \esp\left\{a_{\nr,\nrl}\right\}),\label{eq:hatar}\\
    \hat{\mathbf{b}}_{\nm}=&~\esp\left\{\bmest\right\}+\sqrt{\tau_\UT\rho_\UT}\mathbf{C}_{\bmest,\bmest}\left(\tau_\UT\rho_\UT\mathbf{C}_{\bmest,\bmest}+\mathbf{I}_{\NR}\right)^{-1}\nonumber\\
    &\times\left(\ymtest-\sqrt{\tau_{\UT} \rho_{\UT}} \esp\left\{\bmest\right\}\right),\label{eq:hatbm}
\end{align}
where $\tilde{y}_{\nrl}$ is the $\nrl$th element of $\yrtest$.
\end{proposition} 
\begin{proof}\label{proof1}
 The proof is provided in Appendix A.
\end{proof}
\subsection{Untrusted Data Transmission}
The \tUT~uses the channel estimate obtained in the uplink training phase to 
precode the symbols. Then, during the untrusted data transmission phase, it transmits the precoded signal vector to the \tUR.  Let $\mathbf{x}_{\URUT}\in\mathbb{C}^{\NR\times1}$, with $\esp\left\{\mathbf{x}_{\URUT}\mathbf{x}_{\URUT}^H\right\}=\mathbf{I}_\NR$, be the symbol vector intended to the \tUR. The transmitted signal from  the \tUT, $\mathbf{s}_{\UT} \in \mathbb{C}^{\NT\times 1}$, is written as
\begin{align}\label{eq:st}
    \mathbf{s}_{\UT} = \sqrt{\rho_{\UT}}\mathbf{W}\bm{\Lambda}_\UR^{1/2}\mathbf{x}_{\URUT},
\end{align}
where $\bm{\Lambda}_\UR$ is a diagonal matrix whose diagonal elements are $\lambda_1,\dots,\lambda_\NR$, set to satisfy $\esp\left\{||\mathbf{s}_{\UT}||^2\right\}=\rho_{\UT}$.
Simultaneous to the untrusted data transmission, the MNs operating in the jamming mode send jamming signals to disrupt the communication of the untrusted link.  Let $\mathbf{x}^\mathrm{J}\in\mathbb{C}^{\NR\times1}$, with $\esp\left\{(\mathbf{x}^\mathrm{J})(\mathbf{x}^{\mathrm{J}})^H\right\}=\mathbf{I}_\NR$, denote the jamming symbol vector intended to the \tUR. Then, the signal vector transmitted by the $m$th MN in the jamming mode, $\mathbf{s}_{m}^\mathrm{J} \in\mathbb{C}^{\NM\times1}$, can be written as
\begin{align}\label{eq:sigj}
    \mathbf{s}_{m}^\mathrm{J} = (1-\alpha_m)\sqrt{\rho_\mathrm{J}}\Wj\bm{\Pi}_{m\UR}^{1/2}\mathbf{x}^\mathrm{J},
\end{align}
where $\rho_\mathrm{J}$ is the maximum normalized transmit power at the MNs in jamming mode, while $\alpha_m$ is a binary variable to indicate the operation assignment of each MN $m$, such that  MN $m$ operates in the jamming mode when $\alpha_m=0$  or it operates in the observing mode when $\alpha_m=1$. Also, $\Wj\in\mathbb{C}^{\NM\times \NR}$ is the precoding matrix at the $m$th MN in jamming mode, assumed to follow a maximum-ratio (MR) technique, as it maximizes the jamming power received at \tUR, i.e., $\Wj = \hat{\mathbf{G}}_{m\UR}$.  Note that $\bm{\Pi}_{m\UR}$ is a diagonal matrix with diagonal elements given by $\pi_{m,1},\dots,\pi_{m,\NR}$, chosen to satisfy $\esp\{||\mathbf{s}_{m}^\mathrm{J}||^2\}\leq\rho_\mathrm{J}$  for each MN in jamming mode, which can be further expressed as
\begin{align}~\label{eq:conpi}
    (1-\alpha_m)\sum_{\nr=1}^{\NR}\esp\{||\hat{\mathbf{g}}_{m\UR, \nr}||^2\}\pi_{m\UR,\nr} \leq 1, \forall m.
\end{align}
Therefore, given the transmitted signal $\mathbf{s}_{\UT}$ in \eqref{eq:st} and $\mathbf{s}_{m}^\mathrm{J}$ in \eqref{eq:sigj}, the  signals received at the \tUR~and at the $m$th MN in observing mode are written, respectively, as
\begin{align}
    \mathbf{y}_{\UR} =&~\mathbf{G}_{\UTUR}^H\mathbf{s}_{\UT}+\sum_{m=1}^M\mathbf{G}_{m\UR}^H \mathbf{s}_{m}^\mathrm{J}+\bm{\omega}_{\UR}\nonumber\\
    =&~\sqrt{\rho_{\UT}}\mathbf{G}_{\UTUR}^H\mathbf{W}\bm{\Lambda}_\UR^{1/2}\mathbf{x}_{\URUT}+\sum_{m=1}^M\!(1-\alpha_m)\sqrt{\rho_\mathrm{J}}\nonumber\\
    &\times\mathbf{G}_{m\UR}^H\Wj\bm{\Pi}_{m\UR}^{1/2}\mathbf{x}^\mathrm{J}+\bm{\omega}_{\UR},\label{eq:sigUR}\\
    \mathbf{y}_{m} =&~\alpha_m\mathbf{G}_{\UT m}^H\mathbf{s}_{\UT}+\alpha_m\sum_{m'=1}^M\mathbf{G}_{mm'}^H \mathbf{s}_{m'}^\mathrm{J}+\alpha_m\bm{\omega}_{m}\nonumber\\
    =&~\alpha_m\sqrt{\rho_{\UT}}\mathbf{G}_{\UT m}^H\mathbf{W}\bm{\Lambda}_\UR^{1/2}\mathbf{x}_{\URUT}+\alpha_m\sum_{m'=1}^M\!(1-\alpha_{m'})\sqrt{\rho_\mathrm{J}}\nonumber\\
    &\times\mathbf{G}_{mm'}^H\mathbf{W}_{m'}^{\mathrm{J}}\bm{\Pi}_{m'\UR}^{1/2}\mathbf{x}^\mathrm{J}+\alpha_m\bm{\omega}_{m},\label{eq:sigmt}
\end{align}
where $\bm{\omega}_{\UR}$  and $\bm{\omega}_{m}$ are the $\NR\times1$  and $\NM\times1$ noise vectors at the \tUR~and at the $m$th MN, respectively.  Also, $\mathbf{G}_{mm'}$ denotes the channel matrix between MN $m$ and MN $m'$, modeled as
\begin{align}
    \mathbf{G}_{mm'}&=\sqrt{\beta_{m m'}}\mathbf{H}_{m m'},\label{eq:gmm}
\end{align}
where $\beta_{m m'}$ is the large-scale fading coefficient and $\mathbf{H}_{m m'}$ is the small-scale fading matrix between  MN $m$ and MN $m'$, with i.i.d. $\mathcal{CN}(0,1)$ entries for $m'\neq m$, whereas  $\mathbf{G}_{mm'}=\mathbf{0}$, for $m'=m$.
 
To detect $\mathbf{x}_{\URUT}$, the $m$th MN uses the effective channel estimate  $\hat{\mathbf{B}}_m=[\hat{\mathbf{b}}_{1},\dots,\hat{\mathbf{b}}_{\NM}]^H \in \mathbb{C}^{\NR\times\NR}$  to  combine its  received signal. Specifically,   an MMSE combining matrix is designed as
\begin{align}\label{eq:vmtil}
    {\mathbf{V}}_m=\hat{\mathbf{B}}_m\Big(\hat{\mathbf{B}}_m^H\hat{\mathbf{B}}_m+\varrho\mathbf{I}_{\NR}\Big)^{-1},
\end{align}
where $\varrho$ is the per stream signal-to-noise ratio (SNR). Finally, the aggregated received signal   for observing the untrusted link at the CPU can be obtained as
\begin{align}\label{eq:tildexrt}
    {\mathbf{z}}_{\MN}&=\sum_{m=1}^{M}\alpha_m{\mathbf{V}}_m^H\mathbf{y}_{m}
    =\sum_{m=1}^{M}\alpha_m\left(\mathbf{d}_c+\mathbf{n}_c+\mathbf{i}_c\right),
\end{align}
where
\vspace{-0.2em}
\begin{align}
    \mathbf{d}_{\MN} &\triangleq \sqrt{\rho_{\UT}}{\mathbf{V}}_m^H\mathbf{G}_{\UT m}^H\mathbf{W}\bm{\Lambda}_\UR^{1/2}\mathbf{x}_{\URUT},\label{eq:dc}\\
    \mathbf{n}_{\MN} &\triangleq {\mathbf{V}}_m^H\bm{\omega}_{m},\label{eq:nc}\\
    \mathbf{i}_{\MN} &\triangleq {\mathbf{V}}_m^H\!\sum_{m'=1}^{M}\!(1-\alpha_{m'})\sqrt{\rho_\mathrm{J}}
    \mathbf{G}_{mm'}^H\mathbf{W}_{m'}^{\mathrm{J}}\bm{\Pi}_{m'\UR}^{1/2}\mathbf{x}^\mathrm{J}.\label{eq:jc}
\end{align}
We note that $\mathbf{d}_{\MN}$, $\mathbf{n}_{\MN}$ and $\mathbf{i}_{\MN}$ stand  for the desired signal, the noise, and the inter-MN interference signal, respectively.
\vspace{-0.5em}
\subsection{Large-$M$ Analysis}\label{sec:largem}
 In this subsection, we investigate the monitoring performance of the CF-mMIMO proactive monitoring system in asymptotic regimes, where the number of MNs goes to infinity. To simplify the analysis and gain a better insight, we assume perfect CSI knowledge between the MNs and the \tUT, the MNs, and the \tUR, and between the \tUT~and the \tUR. Moreover, we also consider that the precoding matrix designed by the \tUT~follows the  MRT  precoding scheme, that is, $\mathbf{W}=\mathbf{G}_{\UTUR}$. The analysis for the  ZF precoding technique follows a similar methodology.  Accordingly, we evaluate the asymptotic behavior of the system in two cases: \emph{1}) a large number of MNs in observing mode for a fixed number of MNs in jamming mode, and \emph{2}) a large number of MNs in jamming mode for a fixed number of MNs in observing mode. Let $M_\mathrm{o}$ and $M_\mathrm{J}$ denote the number of MNs in observing mode and in jamming mode, respectively. The corresponding results are presented next in Propositions~\ref{prop:largemo} and~\ref{prop:largemj}.   

\begin{proposition}\label{prop:largemo}
   For any finite $M_{\mathrm{J}}$, as $M_{\mathrm{o}} \to \infty$, and assuming perfect CSI knowledge between all nodes  and MRT precoding at the \tUT, we have 
   \begin{align}\label{eq:asymptmo}
       \frac{{\mathbf{z}}_{\MN}}{M_{\mathrm{o}}} 
       -
       \frac{1}{M_{\mathrm{o}}}\sqrt{\rho_{\UT}}\sum_{m=1}^{M}\alpha_m\esp\{{\mathbf{V}}_m^H\mathbf{G}_{\UT m}^H\mathbf{W}\bm{\Lambda}_\UR^{1/2}\}\mathbf{x}_{\URUT} \xrightarrow{\text{a.s.}} \mathbf{0},
   \end{align}
where $\xrightarrow{\text{a.s.}}$ denotes almost sure (a.s.) convergence.
\end{proposition}
\begin{proof}
    The proof is provided in Appendix B.
\end{proof}
\begin{remark}
Note from \eqref{eq:asymptmo} that as $M_{\mathrm{o}} \to \infty$, the normalized observed signal (normalized by $M_{\mathrm{o}}$) includes only the desired component, as the normalized inter-MN interference and noise are canceled out. Thus, for a finite number of MNs in jamming mode, the monitoring performance is unbounded in terms of the number of MNs in observing mode, and can increase without limit as more observing MNs are employed.
\end{remark}
\begin{proposition}\label{prop:largemj}
    For any finite $M_{\mathrm{o}}$, as $M_{\mathrm{J}} \to \infty$, and assuming perfect CSI knowledge between all nodes and MRT  precoding at the \tUT~and at the MNs in jamming mode, and by down scaling the transmit power of each MN in jamming mode  according to   $\rho_{\mathrm{J}} = \frac{E_{\mathrm{J}}}{M_{\mathrm{J}}^2}$, where $E_{\mathrm{J}}$ is fixed, we have that~\eqref{eq:sigUR} and~\eqref{eq:tildexrt} a.s. converge, respectively, to 
    \begin{align}
        \mathbf{y}_{\UR}&=\sqrt{\rho_{\UT}}\mathbf{G}_{\UTUR}^H\mathbf{W}\bm{\Lambda}_\UR^{1/2}\mathbf{x}_{\URUT}+\bm{\omega}_{\UR}+\frac{\sqrt{E_{\mathrm{J}}}}{M_{\mathrm{J}}}\sum_{m=1}^{M}(1-\alpha_m)\mathbf{G}_{m\UR}^H\nonumber\\&~~~\times\mathbf{G}_{m\UR}\bm{\Pi}_{m\UR}^{1/2}\mathbf{x}^\mathrm{J}-\frac{\sqrt{E_{\mathrm{J}}}}{M_{\mathrm{J}}}\sum_{m=1}^{M}(1-\alpha_m)\esp\{\mathbf{G}_{m\UR}^H\mathbf{G}_{m\UR}\}\nonumber\\&~~~\times\bm{\Pi}_{m\UR}^{1/2}\mathbf{x}^\mathrm{J}\xrightarrow{\text{a.s.}} \mathbf{0}, \text{as } M_{\mathrm{J}} \to \infty,\label{eq:yrmj}\\
        {\mathbf{z}}_{\MN}&\xrightarrow{\text{a.s.}}\sum_{m=1}^{M}\alpha_m{\left(\mathbf{d}_c+\mathbf{n}_c\right)}, \text{as } M_{\mathrm{J}} \to \infty.\label{eq:cpumj}  
    \end{align}
\end{proposition}
\begin{proof}
    For the received signal at the \tUR, by replacing $\rho_{\mathrm{J}}$ as $\frac{E_{\mathrm{J}}}{M_{\mathrm{J}}^2}$, and employing the Chebyshev's inequality, \eqref{eq:yrmj} is attained.  For the received signal at the CPU, we also employ the Chebyshev's inequality. In this case, similar to Proposition \ref{prop:largemo}, note that the inter-MN interference terms   $\mathbf{i}_c$ in ${\mathbf{z}}_{\MN}$  are independent and converge to zero as $M_{\mathrm{J}}\to \infty$, resulting in \eqref{eq:cpumj}.    
\end{proof}
\begin{remark}
 From Proposition \ref{prop:largemj}, note that as $M_{\mathrm{J}} \to \infty$, by setting $\rho_{\mathrm{J}} = \frac{E_{\mathrm{J}}}{M_{\mathrm{J}}^2}$, we have that the inter-MN interference caused by the jamming signal at the CPU is canceled out. Meanwhile, at the \tUR, the effect of the jamming signal is still present. Hence, the monitoring performance can be improved by adjusting $E_{\mathrm{J}}$. 
\end{remark}

\section{Spectral Efficiency}\label{sec:se}
In this section, the SE expressions for the untrusted communication link and the  proactive monitoring system are derived. First, we provide a general formula for the SE with MMSE-SIC scheme given arbitrary side information available at either the \tUR~and MNs, which is independent of the transmit signals.  Next, a closed-form SE expression at the \tUR~is derived, assuming perfect knowledge of the CSI of \tUT. For the proactive monitoring system, two cases are investigated: \emph{1)} imperfect CSI knowledge at the MNs and no CSI
knowledge at the CPU; \emph{2)} imperfect CSI knowledge at both the MNs
and CPU.
\begin{proposition}\label{Prop:SE}
Given the received signal at the \tUR~and at the $m$th MN  as in \eqref{eq:sigUR} and \eqref{eq:sigmt}, respectively, the achievable SE at the \tUR~and at the CPU assuming MMSE-SIC can be written as
\begin{align}
    \mathrm{SE}_\UR&=\Big(\!1\!-\!\frac{\tau_\UT+\tau_\UR}{\tau}\!\Big)\esp\Bigl\{\log_2\left(\det\left(\mathbf{I}_{N_\UR}+\bm{\Upsilon}_\UR\right)\right)\Bigr\},\label{eq:SEr}\\
    \mathrm{SE}_\MN&=\Big(\!1\!-\!\frac{\tau_\UT+\tau_\UR}{\tau}\!\Big)\esp\Bigl\{\log_2\left(\det\left(\mathbf{I}_{N_\UR}+\bm{\Upsilon}_\MN\right)\right)\Bigr\},\label{eq:SEm}
\end{align}
respectively, where $\tau$ is the coherence interval, while $\bm{\Upsilon}_\UR$ and $\bm{\Upsilon}_{\MN}$ are given by
\begin{align}
    \bm{\Upsilon}_\UR&=  \rho_\UT\esp\big\{\bm{\Lambda}_\UR^{1/2}\mathbf{A}_\UR^H|\bm{\Theta}_\UR\big\}(\bm{\Psi}_{\UR})^{-1}\esp\big\{\mathbf{A}_\UR\bm{\Lambda}_\UR^{1/2}|\bm{\Theta}_\UR\big\},\label{eq:upsilonr}\\
    \bm{\Upsilon}_{\MN}&=  \rho_\UT\esp\big\{\mathbf{D}_m^H|\bm{\Theta}_{\MN}\big\}(\bm{\Psi}_{\MN})^{-1}\esp\big\{\mathbf{D}_m|\bm{\Theta}_{\MN}\big\},
\end{align}
where $\bm{\Theta}_{\UR}$ and $\bm{\Theta}_{\MN}$ represent the side information, independent of $\mathbf{x}_{\URUT}$. Moreover, $\mathbf{D}_m$, $\bm{\Psi}_{\UR}$ and $\bm{\Psi}_{\MN}$ are given by  
\vspace{-0.5em}
\begin{align}
    \mathbf{D}_m\triangleq&~\sum_{m=1}^M\alpha_m\mathbf{V}_m^H\mathbf{B}_m\bm{\Lambda}_\UR^{1/2},\\
    \bm{\Psi}_{\UR}=&~\mathbf{I}_{N_\UR}+\rho_{\mathrm{J}}\esp\big\{\Frj(\Frj)^H\big\}+\rho_\UT\esp\big\{\mathbf{A}_\UR\bm{\Lambda}_\UR\mathbf{A}_\UR^H|\bm{\Theta}_\UR\big\}\nonumber\\
    &-\rho_\UT\esp\big\{\mathbf{A}_\UR\bm{\Lambda}_\UR^{1/2}|\bm{\Theta}_\UR\big\}\esp\big\{\mathbf{A}_\UR\bm{\Lambda}_\UR^{1/2}|\bm{\Theta}_\UR\big\}, \\
    \bm{\Psi}_{\MN}=&~\rho_{\mathrm{J}}\esp\left\{\sum_{m=1}^{M}\sum_{l=1}^{M}\alpha_m\alpha_{l}\mathbf{V}_m^H\Fmj_m(\Fmj_l)^H\mathbf{V}_{l}|\bm{\Theta}_{\MN}\right\}\nonumber\\
    &+\rho_\UT\esp\Bigl\{\mathbf{D}_m\mathbf{D}_m^H|\bm{\Theta}_{\MN}\!\Bigr\}\!+\esp\Bigl\{\sum_{m=1}^{M}\sum_{l=1}^{M}\!\!\alpha_m\alpha_l{\mathbf{V}}_m^H{\mathbf{V}}_l|\bm{\Theta}_{\MN}\!\Bigr\}\nonumber\\
    &-\rho_\UT\esp\big\{\mathbf{D}_m|\bm{\Theta}_{\MN}\big\}\esp\big\{\mathbf{D}_m^H|\bm{\Theta}_{\MN}\big\},
\end{align}
respectively, where 
\vspace{-0.5em}
\begin{align}
    \Frj&\triangleq\sum_{m=1}^M\!(1-\alpha_m)\sqrt{\rho_\mathrm{J}}\mathbf{G}_{m\UR}^H\Wj\bm{\Pi}_{m\UR}^{1/2},\\
    \Fmj_m&\triangleq\!\sum_{m'=1}^M(1-\alpha_{m'})\mathbf{G}_{mm'}^H\mathbf{W}_{m'}^{\mathrm{J}}\bm{\Pi}_{m'\UR}^{1/2}.
\end{align}
\end{proposition}
\begin{proof}
   The proof is provided in Appendix C. 
\end{proof}
Now, for the untrusted link, we examine the SE performance when the \tUR~has perfect knowledge of the effective untrusted channel, which represents the worst-case scenario from a monitoring performance perspective. Consequently, we can obtain the following closed-form expression for the SE.
\begin{proposition}
The SE at the \tUR, assuming perfect knowledge of the effective untrusted channel, is given by
\begin{align}\label{eq:ser}
    \mathrm{SE}_\UR=\Big(1-\frac{\tau_\UT+\tau_\UR}{\tau}\Big)\esp\Bigg\{\log_2\Big(1+\sum_{\nr=1}^{\NR}\Gamma_\nr\Big)\Bigg\},
\end{align}
where $\Gamma_\nr$ is the SINR at the $\nr$th receive antenna of \tUR, given by
    \begin{align}\label{eq:sinrnr}
         \Gamma_\nr = \frac{\rho_\UT\lambda_{\nr}|a_{\nr,\nr}|^2}{1+\rho_\UT\lambda_{\nrl}\sum_{\nrl=1 \atop \nrl\neq \nr}^{\NR}|a_{\nr,\nrl}|^2+\rho_{\mathrm{J}}\mathcal{I}},
    \end{align}
    \vspace{-1em}
    with 
    \vspace{-1em}
    \begin{align}
        \mathcal{I}\triangleq&~\NM\sum_{\nrl=1}^\NR\sum_{m=1}^{M}(1-\alpha_m){\pi}_{m,\nrl}\gamma_{m\UR}\nonumber\\
        &+\NM^2\left(\sum_{m=1}^{M}(1-\alpha_m)\sqrt{\pi_{m,\nr}}\gamma_{m\UR}\right)^2.
    \end{align}
\end{proposition}
\begin{proof}
   The proof follows a similar methodology as \cite[Appendix A]{paperzahra} and it is therefore omitted. 
\end{proof}
On the other hand, for the proactive monitoring system, we consider two cases based on the availability of side information, as outlined in the next subsections.\footnote{ In Section \ref{sec:case1}, the superscript (1) stands for the case with imperfect CSI at the MNs and at the CPU, while in Section~\ref{sec:case2}, the superscript (2) stands for the case with imperfect CSI knowledge at the MNs and no CSI knowledge at the CPU.}
\vspace{-0.5em}
\subsection{ Imperfect CSI Knowledge at the MNs and at the CPU}~\label{sec:case1}
In this case, we assume a more centralized architecture where the estimates of the CSI of the \tUT~computed by the MNs in observing mode in \eqref{eq:hatbm} are forwarded to the CPU. Therefore, $\bm{\Theta}_{\MN}^{(1)}=[\hat{\mathbf{B}}_1,\dots,\hat{\mathbf{B}}_{M}]  \in\mathbb{C}^{\NR\times \NR M}$. Note that the elements of $\mathbf{B}_m$ are Gaussian distributed, and thus the MMSE estimates $\hat{\mathbf{B}}_m$ and the corresponding estimation error $\Tilde{\mathbf{B}}_m\triangleq{\mathbf{B}}_m-\hat{\mathbf{B}}_m$ are independent.  Let $\hat{\mathbf{D}}_m\triangleq\sum_{m=1}^M\alpha_m\mathbf{V}_m^H\hat{\mathbf{B}}_m\bm{\Lambda}_\UR^{1/2}$. Hence, using~\eqref{eq:SEm}, the achievable SE at 
the CPU  can be obtained as
\begin{align}\label{eq:se1}
     \mathrm{SE}_{\MN}^{(1)}\!=\!\Big(1-\frac{\tau_\UT+\tau_\UR}{\tau}\Big)\esp\big\{\log_2\big(\det\big(\mathbf{I}_{N_\UR}+\bm{\Upsilon}_{\MN}^{(1)}\big)\big)\big\},
\end{align}
where $\bm{\Upsilon}_{\MN}^{(1)}=\rho_\UT\hat{\mathbf{D}}_m^H(\bm{\Psi}_{\MN}^{(1)})^{-1}\hat{\mathbf{D}}_m$
with
\begin{align}\label{Eq:psi1}
\bm{\Psi}_{\MN}^{(1)}=&~\rho_{\mathrm{J}}\sum_{m=1}^{M}\sum_{l=1}^{M}\alpha_m\alpha_l{\mathbf{V}}_m^H\esp\big\{\Fmj_m(\Fmj_l)^H\big\}{\mathbf{V}}_l\nonumber\\
    &+\sum_{m=1}^{M}\sum_{l=1}^{M}\alpha_m\alpha_l{\mathbf{V}}_m^H{\mathbf{V}}_l+\rho_\UT\esp\big\{\Tilde{\mathbf{D}}_m\Tilde{\mathbf{D}}_m^H\big\},
\end{align}
and  $\Tilde{\mathbf{D}}_m\triangleq\mathbf{D}_m-\hat{\mathbf{D}}_m$.

\vspace{-0.5em}
\subsection{Imperfect CSI Knowledge at the MNs and no CSI Knowledge at the CPU}~\label{sec:case2}
In this case, we assume that the channel estimates computed at each MN are not forwarded to the CPU, hence $\bm{\Theta}_{\MN}^{(2)}=\bm{\varnothing}$. Thus, case-2 models a more lightweight and low-overhead scenario, with \eqref{eq:SEm} being rewritten as
\begin{align}\label{eq:se2}
     \mathrm{SE}_{\MN}^{(2)}=\Big(1-\frac{\tau_\UT+\tau_\UR}{\tau}\Big)\log_2\left(\det\left(\mathbf{I}_{N_\UR}+\bm{\Upsilon}_{\MN}^{(2)}\right)\right),
\end{align}
where $\bm{\Upsilon}_{\MN}^{(2)}=\rho_\UT\esp\big\{\mathbf{D}_m^H\big\}(\bm{\Psi}_{\MN}^{(2)})^{-1}\esp\big\{\mathbf{D}_m\big\}$
with  
\begin{align}\label{Eq:psi2}
    \bm{\Psi}_{\MN}^{(2)}=&~\rho_{\mathrm{J}}\esp\left\{\sum_{m=1}^{M}\sum_{l=1}^{M}\alpha_m\alpha_{l}\mathbf{V}_m^H\Fmj_m(\Fmj_l)^H\mathbf{V}_{l}\right\}\nonumber\\
    &+\rho_\UT\esp\Bigl\{\mathbf{D}_m\mathbf{D}_m^H\Bigr\}+\esp\left\{\sum_{m=1}^{M}\sum_{l=1}^{M}\alpha_m\alpha_l{\mathbf{V}}_m^H{\mathbf{V}}_l\right\}\nonumber\\
    &-\rho_\UT\esp\big\{\mathbf{D}_m\big\}\esp\big\{\mathbf{D}_m^H\big\}.
\end{align}

\begin{remark}
 For case-1, the signaling load required can be quantified as follows: each MN transmits its combined signal to the CPU, $\mathbf{V}_m^H\mathbf{y}_m$, which accounts for a total of $(\tau-(\tau_{\mathrm{t}}+\tau_{\mathrm{r}}))\NR$ complex scalars being transmitted to the CPU by each MN per coherence block. Moreover, since the CPU requires knowledge of the $\NR \times \NR$ complex matrix $\hat{\mathbf{B}}_m$ from the MNs in observing mode, a total of $\NR^2M_{\mathrm{o}}$ statistical parameters are also needed. For case-2, the number of complex scalars to be exchanged between the MNs and CPU per coherence block is also $(\tau-(\tau_{\mathrm{t}}+\tau_{\mathrm{r}}))\NR$. However, as the channel estimates are not forwarded to the CPU, the total number of statistical parameters needed at the CPU is zero~\cite{8845768}. For clarity, these values are tabulated in Table~\ref{tab:sigload}.  
\end{remark}

\begin{table}[t]
\centering
\caption{ Complex scalars and statistical parameters for Case-1 and Case-2}
\begin{tabular}{|c|c|c|}
\hline
\textbf{} & Each coherence block & Statistical parameters \\
\hline
\textbf{Case-1} & $(\tau-(\tau_{\mathrm{t}}+\tau_{\mathrm{r}}))\NR$ & $\NR^2M_{\mathrm{o}}$ \\
\hline
\textbf{Case-2} & $(\tau-(\tau_{\mathrm{t}}+\tau_{\mathrm{r}}))\NR$ & -- \\
\hline
\end{tabular}
\label{tab:sigload}
\end{table}

\vspace{-0.5cm}
\section{Monitoring Success Probability}
 In this section, we start by defining the MSP, which we adopt as the key performance metric to assess the proactive monitoring system.  The MSP quantifies the likelihood that the proactive monitoring system can successfully observe the communication between the UT and UR. Specifically, it is defined as the probability that the SE achieved at the CPU exceeds the SE achieved by the untrusted link. We emphasize that, as discussed in Section~\ref{sec:systemmodel}, the untrusted pair is part of the network infrastructure. Therefore, it is reasonable to assume that the MNs have high-level statistical knowledge of their configuration, which allows the MNs to model the UT-UR effective channel and compute the MSP. Hence, the following indicator function can be employed to denote the event of successful monitoring at the CPU~\cite{paperzahra,7880684}, 
\begin{align}
    \mathrm{MSP}^{(\mathrm{q})}= \begin{cases} 1, & \text{if }\mathrm{SE}_{\MN}^{(\mathrm{q})}\geq\mathrm{SE}_\UR \\ 0, & \text{otherwise},\end{cases}, \mathrm{q} \in \{1,2\},
\end{align}
where $\mathrm{SE}_\UR$, $\mathrm{SE}_{\MN}^{(1)}$, and $\mathrm{SE}_{\MN}^{(2)}$ are given by \eqref{eq:ser}, \eqref{eq:se1} and \eqref{eq:se2}, respectively.
Next, we aim to maximize the MSP performance by optimizing the MN mode assignment coefficients, $\alpha_1, \dots, \alpha_M$, and the jamming power allocation matrices, $\bm{\Pi}_{1\UR}, \dots, \bm{\Pi}_{M \UR}$. Accordingly, let $\mathbf{s}\triangleq[\bm{\alpha},\bm{\pi}_{1\UR},\dots, \bm{\pi}_{M\UR}]$ denote the considered optimization parameters, where $\bm{\alpha}\triangleq[\alpha_1,\dots,\alpha_M]$ and $\bm{\pi}_{m\UR}\triangleq \text{vec}(\bm{\Pi}_{m\UR}), \forall m \in \{1,\ldots, M\}$, with $\text{vec}(\cdot)$ being the vectorization operation. The optimization problem can be formulated as
\begin{subequations}
\begin{align}
\mathcal{P}: \max_{\mathbf{s}} \hspace{3mm} &\mathrm{MSP}^{(\mathrm{q})}, \mathrm{q} \in \{1,2\}~\label{eq:p3} \\
   \text{s. t. }& \eqref{eq:conpi},\nonumber\\
   &\pi_{m\UR,\nr}\geq 0, \forall m, \forall \nr\in\{1, \NR\},\label{eq:bk} \\
   & \alpha_m\in\{0,1\}. \label{eq:maxant}
\end{align}
\end{subequations}
From~\eqref{eq:ser}, \eqref{eq:se1} and \eqref{eq:se2}, note that there is a tight coupling between the optimization variables $\bm{\alpha}$ and $\bm{\pi}_{1\UR}, \dots, \bm{\pi}_{M \UR}$. Moreover, given the dependency between the precoding matrix, $\mathbf{W}$, and the channel matrix between the \tUT~and \tUR, $\mathbf{G}_{\UTUR}$, {the expectations in \eqref{Eq:psi1} and in \eqref{Eq:psi2} cannot be solved or further simplified. Thus, }obtaining a closed-form expression for the MSP is not possible. The difficulty of the problem is further aggravated by its combinatorial nature, given the binary variables $\{\alpha_m\}$ and the continuous parameters $\bm{\pi}_{1\UR}, \dots, \bm{\pi}_{M \UR}$. 
 In such scenarios, gradient-based methods are not suitable as they struggle with the non-differentiability of the objective function and with the discontinuities introduced by binary variables~\cite{tutBayes}. The grid search approach is also not suitable given the computational cost required as the number of MNs increases~\cite{snoek2012practical}. 
Hence, to efficiently address the joint mode selection and jamming power control problem outlined above, we propose an algorithm based on Bayesian optimization. This approach is grounded on the Bayesian inference and thus relies on probability distributions to effectively handle uncertainties in optimization problems, including potential variability within the objective function itself. Bayesian optimization is particularly advantageous in scenarios where the objective function is costly to evaluate,  lacks a known structure like concavity or linearity, or is inherently uncertain, as it builds a probabilistic model to predict promising areas for sampling~\cite{tutBayes}.  For the proposed MSP maximization problem, the Bayesian optimization allows for an analysis with no assumption of factorization or independence between the optimization variables. In the next subsection, we describe the key steps of the Bayesian optimization method.
\vspace{-0.5em}
\subsection{Bayesian Optimization Method}
Bayesian optimization is a sequential model-based method composed of two main components: 1)  a surrogate model, chosen to adequately approximate the objective function in~\eqref{eq:p3} based on a limited number of evaluations, and 2) an acquisition function, used to iteratively update the search in each iteration according to the acquisition function.
\subsubsection{Surrogate Model} 
For the surrogate model, we consider a Gaussian process (GP) regression, which makes use of the mathematical properties of the multivariate normal distribution and can model the behavior of a large variety of functions. Given these characteristics, GPs have been predominately employed in Bayesian optimization algorithms as a prior model of the objective function.  Accordingly, we begin by presenting the formal definition of the GP  below:
 \begin{definition}
 A GP is a collection of random variables, any finite number of which have consistent joint Gaussian distributions~\cite{gpdef}. 
 \end{definition}
Hence, a GP is a nonparametric model, fully characterized by its $\mathcal{S} \times 1$ prior mean function $\bm{\mu}_0$, where $\mathcal{S}$ denotes the search space for optimization parameters $\mathbf{s}$, and by its $\mathcal{S} \times \mathcal{S}$ positive-definite covariance function, also named as kernel, $\mathbf{K}$. Thus, by employing a GP regression, all terms of $\mathbf{\MSP}^{(\mathrm{q})} = [\mathrm{MSP}^{(\mathrm{q})}(\mathbf{s}_1),\dots,\mathrm{MSP}^{(\mathrm{q})}(\mathbf{s}_{N_{\mathrm{opt}}})]$ are assumed to be jointly Gaussian, where $N_{\mathrm{opt}}$ is the number of observations, while the observed outputs $\mathbf{f}= [f_1,\dots,f_{N_{\mathrm{opt}}}]$ are normally distributed given $\mathbf{\MSP}$, that is 
    \begin{align}
        \mathbf{\MSP}^{(\mathrm{q})} | \mathbf{s} &\sim \mathcal{N}(\bm{\mu}_0, \mathbf{K}),  \mathrm{q}\in\{1,2\},\\
        \mathbf{f} | \mathbf{\MSP}^{(\mathrm{q})}, \sigma^2 &\sim \mathcal{N}(\mathbf{\MSP}^{(\mathrm{q})}, \sigma^2\mathbf{I}_{N_{\mathrm{opt}}}),
    \end{align}
where $\sigma^2$ is the observation noise variance. Given a set of observations $\mathcal{D}_{N_{\mathrm{opt}}} = \{\mathbf{s}_i, f_{i}\}_{i=1}^{N_{\mathrm{opt}}}$, the posterior mean and variance can be obtained via Bayes' rule as
\begin{align}
   \mu_{N_{\mathrm{opt}}}(\mathbf{s}) &= \mu_0(\mathbf{s}) + \mathbf{k}(\mathbf{s})^T(\mathbf{K}+\sigma^2\mathbf{I}_{N_{\mathrm{opt}}})^{-1}(\mathbf{f}-\bm{\mu}_0),\label{eq:mukernel} \\
   \sigma_{N_{\mathrm{opt}}}^2(\mathbf{s}) &= k(\mathbf{s},\mathbf{s}) - \mathbf{k}(\mathbf{s})^T(\mathbf{K}+\sigma^2\mathbf{I}_{N_{\mathrm{opt}}})^{-1}\mathbf{k}(\mathbf{s}),\label{eq:sigkernel}
\end{align}
where $\mathbf{k}(\mathbf{s})$ is a vector of covariance terms between $\mathbf{s}$ and $\mathbf{s}_1,\dots, \mathbf{s}_{N_{\mathrm{opt}}}$. The prior mean $\bm{\mu}_0$ provides a possible offset for the samples from the GP, with its value usually set as a constant. On the other hand, the selected kernel $\mathbf{K}$ can highly impact the performance of the GP regression, as it determines the shape of the prior and posterior of the GP. Kernels are mainly classified into two categories: non-stationary and stationary kernels~\cite{7352306}. Non-stationary kernels take into account both the distance between two observed data points and the absolute value of each data point, while stationary kernels depend only on the distance between the observed data points, making them shift-invariant. The stationary kernels can be further classified as an-isotropic and isotropic, with the isotropic kernels being also invariant to rotations~\cite{gpml}. It is important to point out that most stationary kernels depicted in the literature are suitable only for continuous search spaces~\cite{bookBayes}. Thus, to account for the different nature of the optimization variables considered in our optimization problem~\eqref{eq:p3}, i.e., the binary variable $\bm{\alpha}$ for mode assignment and the continuous $\bm{\pi}_{1\UR}, \dots, \bm{\pi}_{M \UR}$ for jamming power control, two separate kernels are considered, denoted as $k_\alpha(\bm{\alpha},\bm{\alpha}')$ and $k_\pi(\bm{\pi}_{m\UR},\bm{\pi}_{m\UR}')$ $\forall m \in \{1, \ldots, M\}$, respectively.
More specifically, $k_\pi(\bm{\pi}_{m\UR},\bm{\pi}_{m\UR}')$  is assumed to follow the Mat\'ern kernel model, a widely used model given its versatility, which is written as 
\begin{align}\label{eq:mk}
    k_\pi(\bm{\pi}_{m\UR},\bm{\pi}_{m\UR}') =&~ \frac{2^{1-\nu}}{\Gamma(\nu)}\left(\sqrt{2\nu}\frac{||\bm{\pi}_{m\UR}-\bm{\pi}_{m\UR}'||}{\iota}\right)\nonumber\\
    &\times B_\nu \left(\sqrt{2\nu}\frac{||\bm{\pi}_{m\UR}-\bm{\pi}_{m\UR}'||}{\iota}\right), 
\end{align}
where $\nu > 0$ is a smoothness parameter, defined as such because it indicates that samples from a GP with Mat\'ern kernel are $[\nu -1]$ times differentiable. Moreover, $\iota$ is the characteristic lengthscale, $\Gamma(\cdot)$ is the Gamma function, and $B_\nu(\cdot)$ is a modified Bessel function of the second kind. For $k_\alpha(\bm{\alpha},\bm{\alpha}')$, we assume a modified Mat\'ern kernel, by considering an integer search space, meaning that $k_\alpha(\bm{\alpha},\bm{\alpha}')$ is modeled as \eqref{eq:mk}. Therefore, $\bm{\alpha}$ is treated as a continuous variable during the optimization process but can only take integer values limited by the constraint~\eqref{eq:maxant}. Accordingly, the overall kernel function is given by 
\begin{align}
    k(\mathbf{s},\mathbf{s}') = k_{\alpha,\pi}\left(\left(\bm{\alpha}, \bm{\pi}_{m\UR}\right),\left(\bm{\alpha}', \bm{\pi}_{m\UR}' \right)\right).
\end{align}




\subsubsection{Acquisition Function}
As described in~\cite{bookBayes}, the selected acquisition function assigns a score to each point in the domain, reflecting the preferences over the locations for the next iteration of the algorithm. Therefore, it should be carefully designed to efficiently explore the search space for promising areas for the next sample of the optimization variables. Given a GP prior, the acquisition functions are commonly based on three parameters: the mean of the optimization variables, the standard deviation of the objective function, and the best-attained value obtained in previous iterations of the optimization algorithm. Traditional acquisition functions are mostly improvement-based (e.g., probability of improvement (PI) and expected improvement (EI)), optimistic-based (e.g., lower confidence bound (LCB)), and information-based (e.g., Thompson sampling (TS) and entropy search (ES))~\cite{brochu}. Nonetheless, it is valid to point out that no single acquisition strategy is capable of providing the best performance in all instances. Hence, a preferred approach would be to change the employed acquisition function throughout the iterations of the optimization process. Accordingly, Hoffman \textit{et al.}~\cite{brochu} proposed an acquisition policy named GP-Hedge, based on the Hedge algorithm in which at each iteration of the optimization algorithm, several acquisition functions deliver candidates for the next sampling points, and a meta-criterion is employed to select the next sampling point among the proposed candidates.


Given the characteristics of the Bayesian optimization described above, the pseudo-code for the proposed algorithm is depicted in \textbf{Algorithm~1}, at the top of the page.

\begin{algorithm}[t]\label{alg:bayesian}
\caption{Bayesian Optimization Algorithm}
\begin{algorithmic}[1]
\footnotesize
\Require{\eqref{eq:p3}, a GP model, acquisition function, number of iterations of the initial phase $N_{\mathrm{initial}}$, number of iterations of the optimization $N_{\mathrm{opt}}$}
\State MSP$^* \gets 0$
\For{$i \gets 1$ to $N_{\mathrm{initial}}$}
\State Generate valid samples of  $\mathbf{S}_i$
\State Compute MSP given the initial $\mathbf{S}_i$
\If{MSP$\geq$  MSP$^*$}
\State $\mathbf{S}^* \gets \mathbf{S}_i$
\State MSP$^* \gets$ MSP
\EndIf
\EndFor
\State Update the GP model with the initial samples  
\For{$i \gets N_{\mathrm{initial}} + 1$ to $N_{\mathrm{opt}}$}
\State Obtain next sampling point according to the acquisition function:
\State $\mathbf{S}_i \gets \mathbf{S}_\mathrm{next}$ 
\State Update the GP model with the updated sampling
\State Compute MSP given $\mathbf{S}_i$
\If{MSP$\geq$  MSP$^*$}
\State $\mathbf{S}^* \gets \mathbf{S}_i$
\State MSP$^* \gets$ MSP
\EndIf
\EndFor
\Ensure{$\mathbf{S}^*$, MSP$^*$}  
\end{algorithmic}
\end{algorithm}
\vspace{-0.5em}
\subsection{Complexity and Convergence Analysis}
The main computational complexity of \textbf{Algorithm~1} comes from updating the GP-based surrogate model, with the order of $\mathcal{O}(N_{\mathrm{opt}}^3)$~\cite{7352306}. The cost is mainly due to the inversion of the covariance matrix in \eqref{eq:mukernel} and \eqref{eq:sigkernel}. It is worthwhile to point out that in most Bayesian optimization-based schemes, including the one proposed in \textbf{Algorithm~1}, the number of observations $N_{\mathrm{opt}}$ is small. Nonetheless, several attempts have been proposed in the literature to reduce the computational burden of updating the surrogate model for cases when $N_{\mathrm{opt}}$ is large. Options to reduce the computational cost include sparsification techniques for  GPs, such as sparse pseudo-input GPs~\cite{NIPS2005_4491777b} and sparse spectrum GPs~\cite{JMLR:v11:lazaro-gredilla10a}.

The convergence of the Bayesian optimization-based algorithm primarily depends on the chosen acquisition function and is evaluated based on Bayesian regret. As described in~\cite{brochu}, the GP-hedge approach complicates the assessment of the convergence behavior of \textbf{Algorithm~1}, as decisions made at each iteration 
 influence the problem state and the selection criterion for the acquisition function in all subsequent iterations. Therefore, the regret attained with \textbf{Algorithm~1} cannot be directly related to the regret of one of the acquisition functions considered by the GP-hedge approach. Under these considerations, in \cite[Theorem 1]{brochu}, a bound on the cumulative regret of the Bayesian optimization algorithm employing GP-Hedge was derived which is displayed below for the sake of completeness:
\vspace{-0.5em}
    \begin{align}
        R_{N_{\mathrm{opt}}}\leq \sqrt{{N_{\mathrm{opt}}} C_1\gamma_{N_{\mathrm{opt}}}}+\sum_{i=1}^{N_{\mathrm{opt}}}  \sigma_{i-1}\mathbf{s}_i+\mathcal{O}(\sqrt{{N_{\mathrm{opt}}}}),
    \end{align}
    where $\gamma_{N_{\mathrm{opt}}}$ is a bound on the information gained at points selected by the algorithm after $N_{\mathrm{opt}}$ iterations, whereas $C_1 = 2/\log(1+\sigma^{-2})$, and $\mathbf{s}_i$ is the $i$th optimization vector proposed by the proposed algorithm. 

\section{Numerical Results}\label{sec:Results}
In this section, numerical results are presented to exploit the performance of our CF-mMIMO proactive monitoring system under different available CSI scenarios, and to evaluate the performance of the proposed Bayesian optimization-based algorithm to maximize the MSP of  CF-mMIMO-based proactive monitoring.
\vspace{-1em}
\subsection{Simulation Setup and Parameters}
We consider a CF-mMIMO proactive monitoring system where the MNs, the \tUT, and the \tUR~are uniformly distributed within a $\mathrm{D}\times \mathrm{D}$ km$^2$ area. The wrapped-around technique is used to avoid the boundary effects. The channel bandwidth is set to $B = 20$ MHz, the pilot length of $\tau_\UR$ and $\tau_\UT$ are set to 40, and the coherence interval is $\tau = 300$. The maximum transmit powers at the \tUR, \tUT, and  MNs in jamming mode are set as $100$ mW, $100$ mW, and $200$ mW, respectively. The corresponding normalized transmit powers $\rho_\UR$, $\rho_\UT$, and $\rho_\mathrm{J}$, are calculated by dividing the maximum transmit power by the noise power, given by
\begin{align}
    \text{noise power} = B \times k_B \times T_0 \times \text{noise figure }  (\mathrm{W}),
\end{align}
where $k_B = 1,381 \times 10^{-23}$ (Joule per Kelvin) is the Boltzmann constant, $T_0 = 290$ (Kelvin) is the noise temperature, and the noise figure is 9 dB. The large-scale fading coefficients $\beta_{m\UR}$, $\beta_{m\UT}$, $\beta_{\UTUR}$ and $\beta_{mm'}$ are computed as
\begin{align}
    \beta_{\mathrm{l}} = \mathrm{PL}_{\mathrm{l}} 10^{\frac{\sigma_{\mathrm{sh}}z_{\mathrm{l}}}{10}}, ~\mathrm{l} \in \{m\UR, m\UT, \UTUR, mm'\},
\end{align}
where $\mathrm{PL}_{\mathrm{l}}$ represents the path loss, and $10^{\frac{\sigma_{\mathrm{sh}}z_{\mathrm{l}}}{10}}$ represents the shadow fading with the standard deviation $\sigma_{\mathrm{sh}} = 8$ dB, and $z_{\mathrm{l}} \sim \mathcal{N}(0, 1)$. Based on~\cite{7827017}, let $d_{\mathrm{l}}$ denote the distance between the nodes, such that the path loss $\mathrm{PL}_{\mathrm{l}}$ is computed as
\begin{align}
    \mathrm{PL}_{\mathrm{l}}= \begin{cases}-L-35 \log _{10}\left(d_{\mathrm{l}}\right), & d_{\mathrm{l}}>d_1 \\ -L-15 \log _{10}\left(d_1\right)-20 \log _{10}\left(d_{\mathrm{l}}\right), & d_0<d_{\mathrm{l}} \leq d_1 \\ -L-15 \log _{10}\left(d_1\right)-20 \log _{10}\left(d_0\right), & d_{\mathrm{l}} \leq d_0,\end{cases}
\end{align}
\vspace{-0.5em}
where
\begin{align}
L &\triangleq~46.3+33.9 \log _{10}(f)-13.82 \log _{10}\left(h_{\mathrm{MN}}\right) \nonumber\\
& -\left(1.1 \log _{10}(f)-0.7\right) h_{\mathrm{u}}+\left(1.56 \log _{10}(f)-0.8\right),
\end{align}
with $f = 1.9$ GHz being the carrier frequency, $h_{\mathrm{MN}} = 15$ m and $h_{\mathrm{u}} = 1.65$ m are the antenna heights of the MNs and of the untrusted users, respectively. Moreover, $d_0 = 10$ m and $d_1 = 50$ m. The per stream SNR in \eqref{eq:vmtil} is set as $\varrho=1/\rho_\UT$, while ZF and MRT precoding techniques are considered for the transmit precoding matrix $\mathbf{W}$ at \tUT, that is,
\begin{align}
    \mathbf{W} &= \hat{\mathbf{G}}_{\UTUR}(\hat{\mathbf{G}}_{\UTUR}^H\hat{\mathbf{G}}_{\UTUR})^{-1}, ~ \text{for ZF,}\\
    \mathbf{W} &= \hat{\mathbf{G}}_{\UTUR}, ~ \text{for MRT}.
\end{align}
Further, unless stated otherwise, $\mathrm{D} = 1$ Km, $M = 8$, $N = 30$, and $N_\UT = N_\UR = 4$. Finally, regarding the Bayesian optimization-based algorithm, the simulations are built using the  \textit{gp\_minimize} library of Python and the results are obtained on a Dell laptop with Intel Core$^{\mathrm{TM}}$ $i7-1365U$ and RAM of $16$ GB. The number of iterations of the initial phase is $N_{\mathrm{initial}} = 10$, while the number of iterations of the optimization is $N_{\mathrm{opt}} = 20$.

Next, we evaluate the performance of the proposed CF-mMIMO proactive monitoring system relying on the optimized MSP with \textbf{Algorithm~1}  under two cases of CSI availability at the MNs and CPU: \textbf{Case-1}: imperfect CSI knowledge at both the MNs and CPU, and \textbf{Case-2}: imperfect CSI knowledge at the MNs with no CSI knowledge at the CPU.
 
\subsection{Performance Evaluation}

\subsubsection{Comparison between CF-mMIMO and Co-located mMIMO}
\begin{figure}[t]
\centering
  \begin{subfigure}[t]{0.45\textwidth}
    \includegraphics[width=\textwidth]{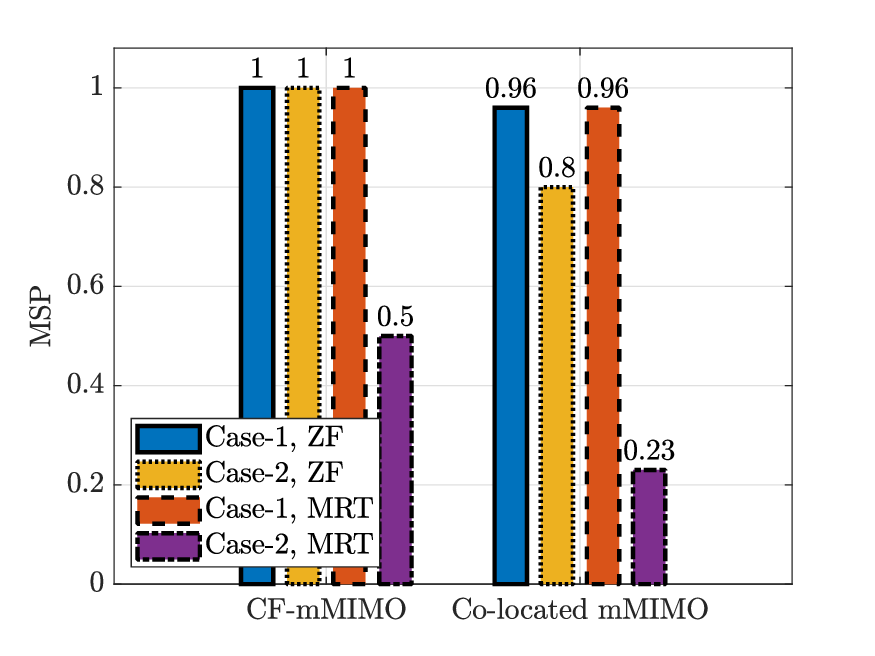}
    \caption{$\mathrm{D} = 0.5$ Km}
    \label{fig:MSPvsDzf}
  \end{subfigure}
  \begin{subfigure}[t]{0.45\textwidth}
    \includegraphics[width=\textwidth]{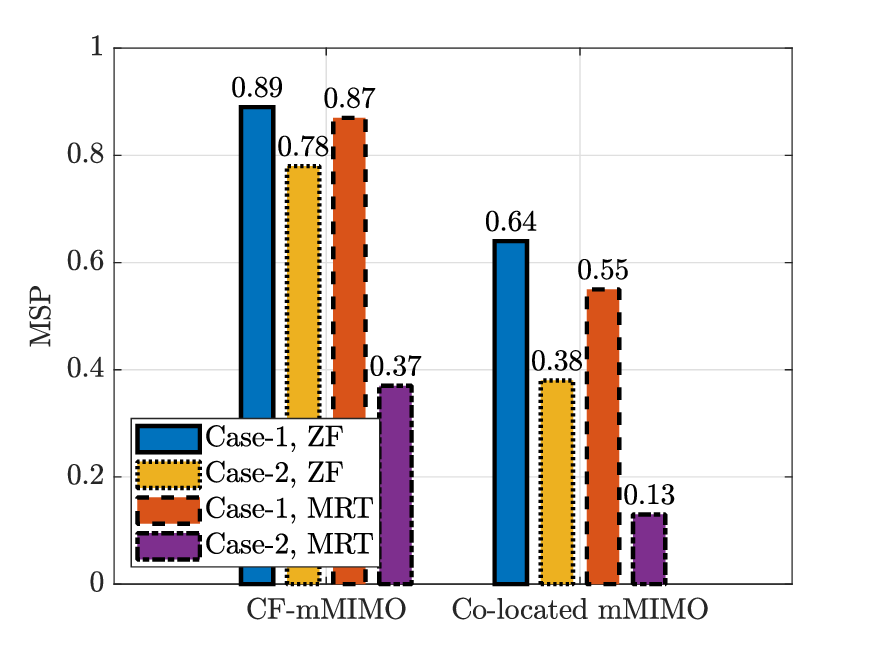}
    \caption{$\mathrm{D} = 1.5$ Km}
    \label{fig:MSPvsDmrt}
  \end{subfigure}
  \caption{MSP versus the area size, $\mathrm{D}$ (Km). }
\vspace{-0.52cm}
\label{fig:MSPvsD}
\end{figure}

Figure~\ref{fig:MSPvsD} illustrates the MSP performance of the  CF-mMIMO proactive monitoring system relying on our proposed CSI acquisition scheme and joint MNs mode assignment and power control optimization in \textbf{Algorithm~1} versus the area size, $\mathrm{D}$.   For comparison, we present results for a co-located mMIMO-aided proactive monitoring system operating in FD mode, where all MNs are arranged in a co-located antenna array and simultaneously perform observation and jamming. For the co-located mMIMO system, a residual self-interference level of $30$ dB is assumed.\footnote{ After employing self-interference mitigation techniques, the strength of the residual self-interference in co-located mMIMO system can range from 30 to 100 dB~\cite{10258345}.} Moreover, half of the antennas at the co-located MNs are dedicated to jamming transmissions, while the remaining antennas are employed to observe the untrusted communication link.  For a fair comparison, the jamming power allocation for the co-located mMIMO system is also optimized with the proposed Bayesian optimization approach. Specifically, in Fig.~\ref{fig:MSPvsD}(a), the results are presented considering an area size of $\mathrm{D}= 0.5$ Km, while   Fig.~\ref{fig:MSPvsD}(b) displays the results with an area size of $\mathrm{D}= 1.5$ Km. Notably, the CF-mMIMO proactive monitoring system outperforms the co-located mMIMO system, irrespective of the size of the area, the CSI availability at the CPU, or the selected precoding scheme at the \tUT. The performance gap between them widens as the area size increases. For instance, for $\mathrm{D} = 0.5$ Km, the CF-mMIMO system provides around $5$\%  and $20$\% improvement in the MSP performance in comparison to the co-located system for \textbf{Case-1} and \textbf{Case-2}, respectively. For $\mathrm{D} = 1.5$ Km, the corresponding MSP performance gains increase to $25$\% and $40$\%, respectively. This is explained by the macro-diversity attained with the CF-mMIMO system, with multiple distributed MNs surrounding the \tUT~and the \tUR~in contrast to the co-located system, where no macro-diversity gain is achieved. Moreover, the monitoring performance of the CF-mMIMO system is better when the \tUT~utilizes the ZF precoding scheme compared to when the MRT scheme is employed.

\subsubsection{Comparison Between Algorithm~1 and Benchmark Schemes}

\begin{figure}[t]
\centering
  \begin{subfigure}[t]{0.45\textwidth}
    \includegraphics[width=\textwidth]{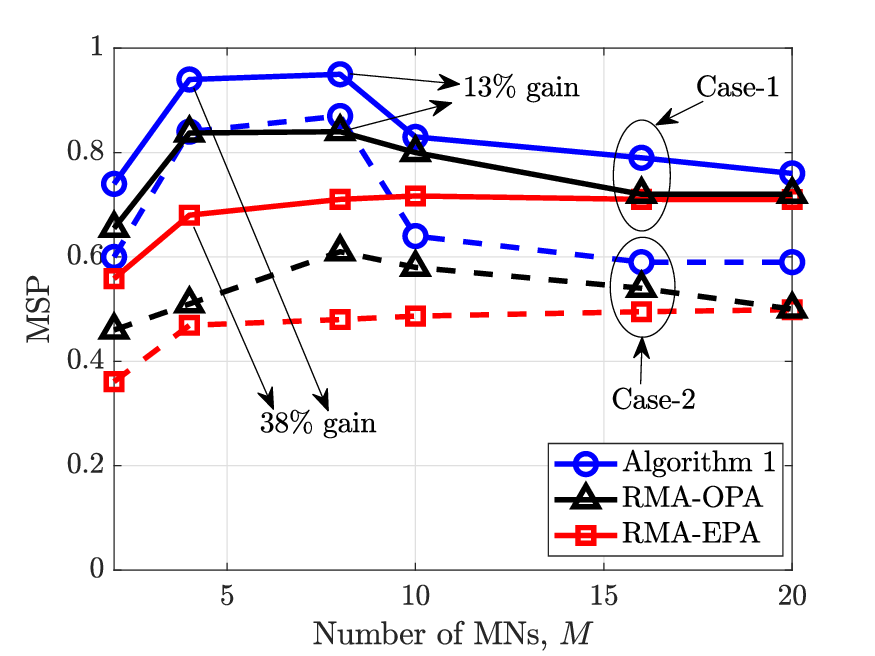}
    \caption{ZF precoding}
    \label{fig:MSPvsMzf}
  \end{subfigure}
  \begin{subfigure}[t]{0.45\textwidth}
    \includegraphics[width=\textwidth]{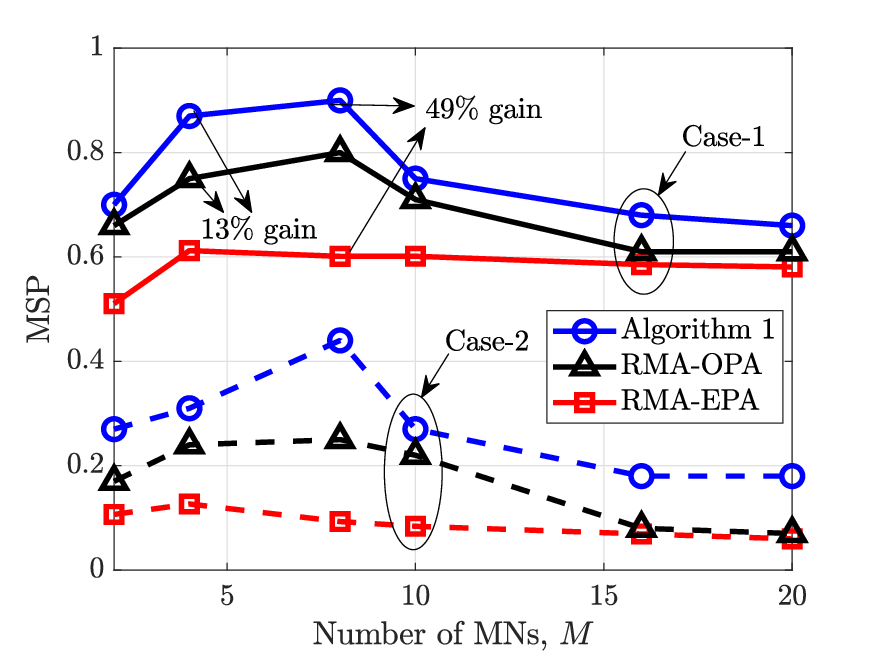}
    \caption{MRT precoding}
    \label{fig:MSPvsMmrt}
  \end{subfigure}
  \caption{ MSP versus the number of MNs, $M$, with the total number of antennas for all the MNs fixed as 240.}
\vspace{-0.52cm}
\label{fig:MSPvsM}
\end{figure}

In Fig.~\ref{fig:MSPvsM}, we evaluate the MSP performance of a CF-mMIMO proactive monitoring system as a function of the number of MNs, with the total number of antennas across all MNs fixed as $240$. 
In this figure, we compare our optimized approach in \textbf{Algorithm~1} with {the random mode assignment and optimized power allocation (RMA-OPA) scheme, where the MNs are randomly assigned to either observing mode or jamming mode, whereas the jamming power control is optimized according to \textbf{Algorithm~1}. We also depict RMA and equal power allocation (RMA-EPA), where the jamming power in \eqref{eq:conpi} is distributed equally across all MNs in jamming mode. We observe that our proposed joint mode assignment and jamming power control with \textbf{Algorithm~1}  shows the best MSP performance over all evaluated schemes, followed by the RMA-OPA scheme. Both schemes significantly enhance the monitoring performance compared to the RMA-EPA in all cases. Accordingly, even though the optimized scheme has a complexity of $\mathcal{O}(N_{\mathrm{opt}}^3)$, it achieves over 38\% and 49\% MSP performance gain in comparison to the RMA-EPA counterpart for ZF and MRT precoding designs at the UT, respectively. Notably, as the number of MNs in the proactive monitoring system increases, the MSP remains relatively stable under the RMA-EPA scheme. In contrast, for the optimized and for the RMA-OPA, the MSP is maximized when $M = 8$.  Moreover, note that the setup considered in Fig.~\ref{fig:MSPvsM} is different to the one evaluated in Section~\ref{sec:largem}. In Section~\ref{sec:largem}, we aim to evaluate the asymptotic behavior of the system as we increase the number of MNs in either observing or jamming mode, for a fixed number of MNs for the non-evaluated mode. Meanwhile in Fig.~\ref{fig:MSPvsM}, the number of MNs in observing and jamming mode increase simultaneously. Hence, the obtained results here can be explained by the fact that for a larger number of MNs, the macro-diversity gain increases and the path-loss decreases. On the other hand, since the total number of antennas at the MNs is fixed, as $M$ gets larger, $N$ reduces, resulting in a smaller array gain. In particular, note that the performance of the RMA-OPA scheme is closer to that of the RMA-EPA scheme for $M \geq 16$, which emphasizes that the system benefits from the macro-diversity gain achieved with the optimized mode assignment. For the RMA-EPA design, the benefit of increasing the macro-diversity gain is consistently dominant, whereas for the optimized design, the trade-off between increasing the macro-diversity and decreasing the array gain has a clear impact on MSP performance.}

\subsubsection{Comparison Between the Perfect CSI Scenario at the MNs and CPU, and Case 1 and Case 2}

\begin{figure}[t]
    \centering
    \includegraphics[width=0.45\textwidth]{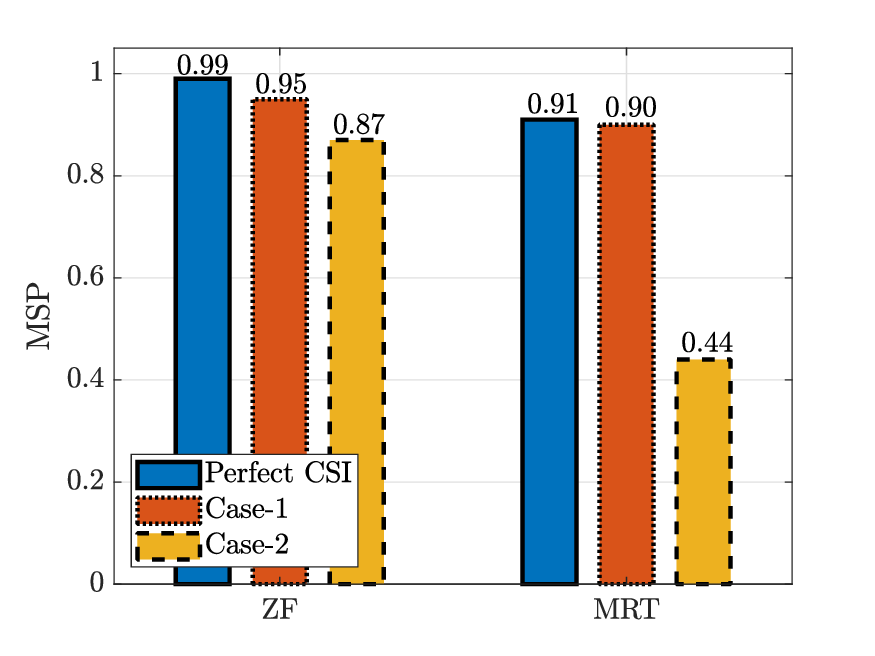}
    \caption{ MSP performance for different cases of CSI availability.}
    \label{fig:MSPvscases}
    \vspace{-0.52cm}
\end{figure}

 Figure~\ref{fig:MSPvscases} illustrates the MSP performance for the different  CSI availability cases and precoding schemes at \tUT. For performance comparison, we also illustrate the results for the ideal monitoring scenario, in which perfect CSI knowledge of the untrusted communication link is available at the MNs and at the CPU. As expected, the perfect CSI case presents an upper bound on the MSP performance. Nevertheless, the performance gap between the perfect CSI case and \textbf{Case-1} is small for either precoding scheme at the \tUT. This result highlights the effectiveness of our proposed acquisition approach. Moreover, it also emphasizes that, even though there is a number of statistical parameters that needs to be transmitted to the CPU in \textbf{Case-1} in comparison to \textbf{Case-2}, the performance is nearly optimal for the former. Therefore, the added complexity is justified by the performance gain achieved with \textbf{Case-1}.

\subsubsection{Impact of the Number of Antennas at the MNs}

\begin{figure}[t]
    \centering
    \includegraphics[width=0.45\textwidth]{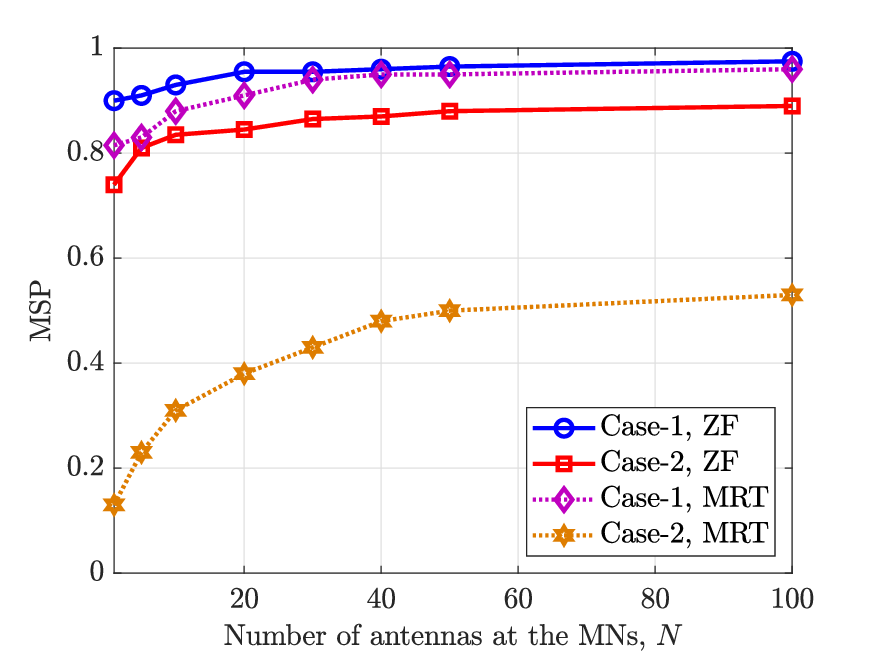}
    \caption{MSP  versus the number of antennas at the MNs, $\NM$.}
    \label{fig:MSPvsN}
\end{figure}

Figure~\ref{fig:MSPvsN} illustrates the MSP as a function of the number of antennas at the MNs, $N$, for the proposed joint mode assignment and jamming power control optimization in \textbf{Algorithm~1}. For \textbf{Case-1}, increasing  $N$ from $1$ to $30$ yields an MSP improvement of approximately 
$8\%$ for the ZF precoding scheme and  $15\%$ for the MRT precoding scheme. For both schemes in \textbf{Case-1}, the MSP stabilizes for $N$ above $30$. Conversely, in the absence of CSI at the CPU, i.e., \textbf{Case-2}, the MSP performance gain rises to around  $15\%$ for the ZF precoding and  $40\%$ for the MRT precoding with $N$ ranging from 1 to 50.  A similar asymptotic trend is observed when  $N$ exceeds $50$. Figure~\ref{fig:MSPvsN} further demonstrates that the influence of increasing  $N $ is more pronounced in the CF-mMIMO system when CSI is unavailable.  

\subsubsection{Impact of the Number of Antennas at the \tUR}

\begin{figure}[t]
    \centering
    \includegraphics[width=0.45\textwidth]{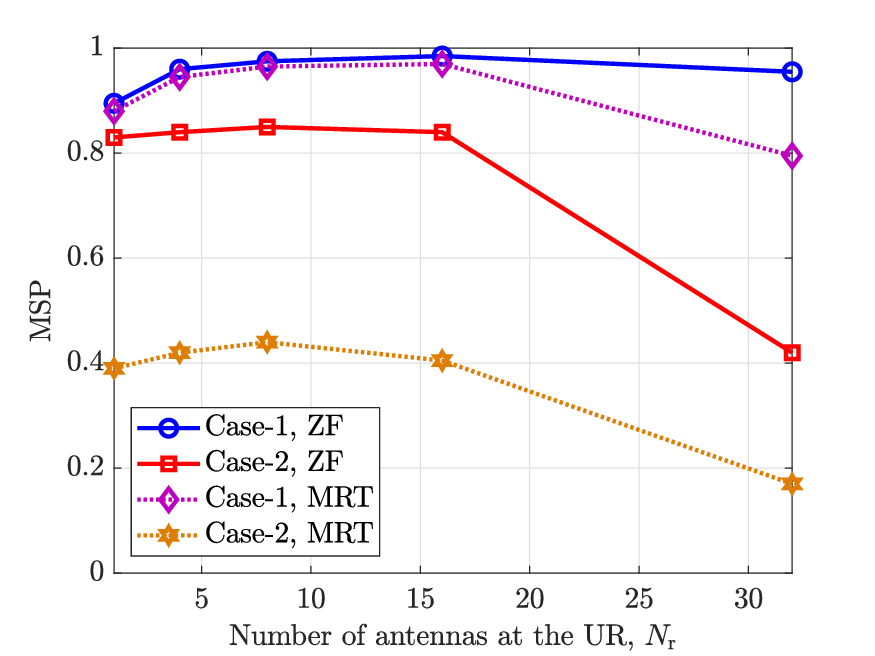}
    \caption{MSP  versus the number of antennas at the \tUR, $\NR$, with $\NT=\NR$.}
    \label{fig:MSPvsNr}
    \vspace{-0.52cm}
\end{figure}

Figure~\ref{fig:MSPvsNr} illustrates the MSP versus the number of antennas at the \tUR, $\NR$, assuming an equal number of antennas at the \tUT, that is $\NT=\NR$. We observe that there is a slight performance improvement for all the evaluated cases as the number of antennas at the untrusted pair increases from a range of $\NR=1$ to $\NR=16$. This is due to the fact that as the number of antennas at the untrusted pair increases, the quality of the channel estimations performed at the uplink training and beamforming training phases is enhanced. Moreover, based on the SINR expression for the \tUR, given in \eqref{eq:sinrnr},  as the number of antennas at the \tUR~gets larger, the received signal strength from the \tUT~improves. However, this also leads to an increase in interference from other antennas and from the jamming signal transmitted by the MNs in the jamming mode at the UR. On the other hand, for $\NR$ greater than $16$, the monitoring performance of all cases decreases, which is expected given that the numerator in \eqref{eq:sinrnr} becomes significantly larger than the interference terms in the denominator. Therefore, the SINR at the \tUR~increases and hence, MSP decreases. Interestingly, with imperfect CSI knowledge at both MSP and CPU, the MSP performance of our CF-mMIMO proactive monitoring system remains high (greater than $0.8$), regardless of the number of antennas at the untrusted nodes or the precoding scheme for the untrusted transmission. This result highlights the effectiveness of our proposed CSI acquisition scheme and joint optimization approach.

\subsubsection{Impact of the Transmit Power at the MNs in Jamming Mode}

\begin{figure}[t]
    \centering
    \includegraphics[width=0.45\textwidth]{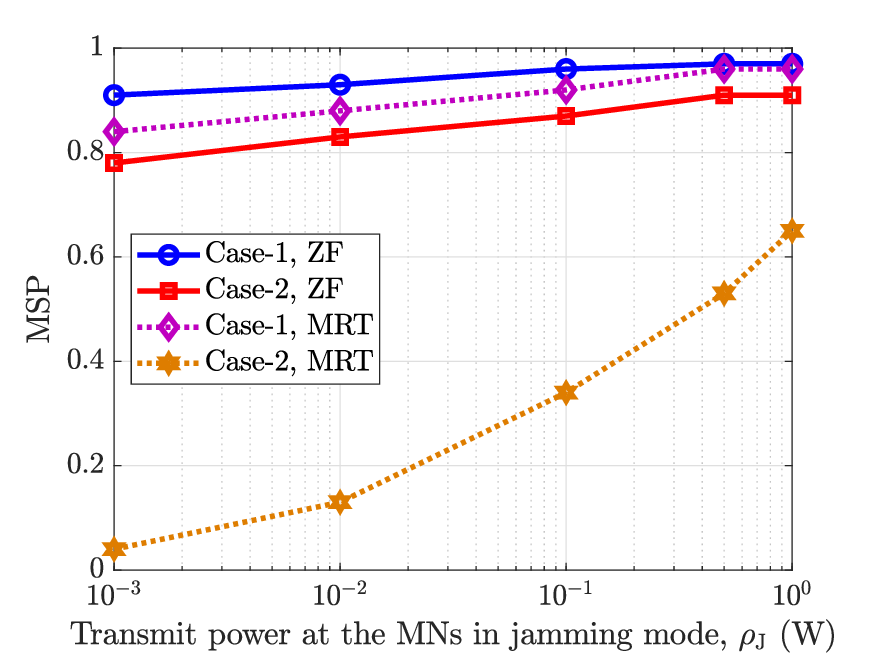}
    \caption{ MSP  versus the transmit power at the MNs in jamming mode, $\rho_{\mathrm{J}}$ (W).}
    \label{fig:MSPvspowerm}
    \vspace{-0.52cm}
\end{figure}

 Finally, Fig.~\ref{fig:MSPvspowerm} shows the MSP as a function of the transmit power at the MNs in jamming mode, $\rho_{\mathrm{J}}$.  Note that increasing $\rho_{\mathrm{J}}$ yields a better performance for all evaluated cases. Precisely, for \textbf{Case-1}, changing $\rho_{\mathrm{J}}$ from 1 mW to 1 W, provides a MSP performance gain of 7\% for the ZF precoding scheme, and of 12\% for the MRT precoding scheme. On the other hand, for \textbf{Case-2}, the MSP improvement is of 13\% for the ZF precoding scheme and of 61\% for the MRT precoding for the same range of $\rho_{\mathrm{J}}$. Note that for both schemes in \textbf{Case-1} and for \textbf{Case-2} with ZF precoding, the MSP stabilizes for a $\rho_{\mathrm{J}}$ value larger than 200 mW.

\section{Conclusions}\label{sec:Conclusions}
We exploited the MSP performance of a CF-mMIMO proactive monitoring system with multi-antenna MNs designed to monitor a multi-antenna untrusted pair, utilizing an efficient CSI acquisition approach. To optimize the monitoring performance, we developed a joint mode assignment and jamming power control algorithm using the Bayesian optimization approach. The MSP performance of the CF-mMIMO proactive monitoring system was evaluated under two CSI availability cases. Our findings demonstrated that the proposed joint optimization algorithm provides significant MSP performance gains over the RMA-EPA approach, particularly under conditions of imperfect CSI at both the MNs and CPU. Furthermore, the results confirmed that the MSP performance of the CF-mMIMO proactive monitoring system remains robust across various network configurations, showing resilience to different precoding schemes for the untrusted transmission, and also to variations in the number of antennas at the untrusted nodes. Moreover,  we validated that the CF-mMIMO proactive monitoring system, coupled with our CSI acquisition scheme and optimization algorithm, substantially outperforms the co-located mMIMO monitoring system.
 Potential future works could include a CSI acquisition scheme for the channel between the MNs. With this, the inter-MN interference can be reduced. Nonetheless, that scenario could entail new security risks, as the untrusted pair may attempt to estimate its channel with the MNs as well.  
\vspace{-0.3cm}
\section*{Appendix A \\Proof of Proposition~1}
The MMSE estimate of $\bmest$ can be computed as
\begin{align}\label{eq:bma}
\hat{\mathbf{b}}_{\nm}&=\esp\bigl\{\bmest\bigr\}+\mathbf{C}_{\bmest,\ymtest}\mathbf{C}_{\ymtest\ymtest}^{-1}\left(\ymtest-\esp\bigl\{\ymtest\bigr\}\right),
\end{align}
where $\mathbf{C}_{\bmest,\ymtest}$ is given by
\begin{align}\label{eq:cbm1}
    \mathbf{C}_{\bmest,\ymtest}=\esp\left\{\left(\bmest-\esp\bigl\{\bmest\bigr\}\right)\left(\ymtest-\esp\bigl\{\ymtest\bigr\}\right)^T\right\}. 
\end{align}
Replacing \eqref{eq:ynmd} into \eqref{eq:cbm1}, and given that $\omegam$ has i.i.d. $\mathcal{CN}(0,1)$ entries and is independent of $\bmest$, \eqref{eq:cbm1} can be rewritten as
\begin{align}\label{eq:cbm2}
   \mathbf{C}_{\bmest,\ymtest} &=\sqrt{\tau_\UT\rho_\UT}\esp\Bigl\{\bmest\bmest^T\!-\!\bmest\esp\bigl\{\bmest^T\bigr\}\!-\!\esp\{\bmest\}\bmest^T\nonumber\\&+\!\esp\bigl\{\bmest\bigr\}\esp\bigl\{\bmest^T\bigr\}\Bigr\}=\sqrt{\tau_\UT\rho_\UT}\mathbf{C}_{\bmest,\bmest}.
\end{align}
Analogously to \eqref{eq:cbm1}, $\mathbf{C}_{\ymtest,\ymtest}$ can be rewritten as
\begin{align}\label{eq:cym}
     \mathbf{C}_{\ymtest,\ymtest}&=\esp\left\{\left(\ymtest-\esp\bigl\{\ymtest\bigr\}\right)\left(\ymtest-\esp\bigl\{\ymtest\bigr\}\right)^T\right\} \nonumber\\
    &=\tau_\UT\rho_\UT\esp\Bigl\{\bmest\bmest^T\!-\!\bmest\esp\bigl\{\bmest^T\bigr\}\!-\!\esp\bigl\{\bmest\bigr\}\bmest^T\nonumber\\
    &+\!\esp\bigl\{\bmest\bigr\}\esp\bigl\{\bmest^T\bigr\}\Bigr\}\!+\!\esp\Bigl\{\omegam\omegam^T\!\!\!+\!\esp\bigl\{\omegam\bigr\}\esp\bigl\{\omegam^T\bigr\}\Bigr\}\nonumber\\
    &=\tau_\UT\rho_\UT\mathbf{C}_{\bmest,\bmest}\!+\!\mathbf{I}_{\NR}.
\end{align}
By substituting \eqref{eq:cbm2} and \eqref{eq:cym} into \eqref{eq:bma}, we obtain \eqref{eq:hatbm}.
\section*{Appendix B \\Proof of Proposition~\ref{prop:largemo}}
As the signal components in~\eqref{eq:dc},~\eqref{eq:nc}, and~\eqref{eq:jc} are independent, the convergence of each of them can be evaluated separately. First, by using the law of large numbers and the Chebyshev's inequality, we have that the noise, $\mathbf{n}_\MN$,  and inter-MN interference, $\mathbf{i}_\MN$, converge as
\begin{align}
    &\frac{1}{M_{\mathrm{o}}}\sum_{m=1}^{M}\alpha_m\left(\mathbf{n}_{\MN}-\esp\{\mathbf{n}_{\MN}\}\right) \xrightarrow{\text{a.s.}} \mathbf{0}, \text{as } M_{\mathrm{o}}\to \infty,\label{eq:nmoap}\\
    &\frac{1}{M_{\mathrm{o}}}\sum_{m=1}^{M}\alpha_m\left(\mathbf{i}_{\MN}-\esp\{\mathbf{i}_{\MN}\}\right)\xrightarrow{\text{a.s.}} \mathbf{0}, \text{as } M_{\mathrm{o}}\to \infty.\label{eq:intmoap}
\end{align}
Given that the terms in $\mathbf{n}_\MN$ and $\mathbf{i}_\MN$ are independent, and knowing that $\bm{\omega}_m$ and $\mathbf{G}_{mm'}$ contain i.i.d. variables with zero mean, the inter-MN interference and noise are canceled out. On the other hand, for the desired signal, we also employ the Chebyshev's inequality as follows: 
\begin{align}
    &\frac{1}{M_{\mathrm{o}}}\sqrt{\rho_{\UT}}\left(\sum_{m=1}^{M}\alpha_m{\mathbf{V}}_m^H\mathbf{G}_{\UT m}^H\mathbf{W}\bm{\Lambda}_\UR^{1/2}\mathbf{x}_{\URUT}\right.\nonumber\\ 
    &\left.-\sum_{m=1}^{M}\alpha_m\esp\{{\mathbf{V}}_m^H\mathbf{G}_{\UT m}^H\mathbf{W}\bm{\Lambda}_\UR^{1/2}\}\mathbf{x}_{\URUT}\right) \xrightarrow{\text{a.s.}} \mathbf{0}, \text{as } M_{\mathrm{o}}\to \infty.\label{eq:dsmoa}
\end{align}
In this case, we need to certify that the expectation term does not goes to 0 as $M_{\mathrm{o}}\to \infty$. Thus, assuming perfect channel estimation and MRT precoding, \eqref{eq:dsmoa} is rewritten as
\begin{align}
    &\frac{1}{M_{\mathrm{o}}}\sqrt{\rho_{\UT}}\left(\sum_{m=1}^{M}\alpha_m{\mathbf{V}}_m^H\mathbf{G}_{\UT m}^H\mathbf{G}_{\UTUR}\bm{\Lambda}_\UR^{1/2}\mathbf{x}_{\URUT}\right.\nonumber\\ 
    &\left.-\sum_{m=1}^{M}\alpha_m\esp\{{\mathbf{V}}_m^H\mathbf{G}_{\UT m}^H\mathbf{G}_{\UTUR}\bm{\Lambda}_\UR^{1/2}\}\mathbf{x}_{\URUT}\right) \xrightarrow{\text{a.s.}} \mathbf{0}, \text{as } M_{\mathrm{o}}\to \infty.
\end{align}
Given that $\mathbf{G}_{\UT m}$ and $\mathbf{G}_{\UTUR}$ are independent, and modeled, respectively as~\eqref{eq:gtm} and \eqref{eq:channeldef}, it is straightforward to see that the expectation term does not converges to 0 as $M_{\mathrm{o}}\to \infty$, and hence \eqref{eq:asymptmo} is obtained.  


\section*{Appendix C \\Proof of Proposition~\ref{Prop:SE}}
By denoting the differential entropy as $h(\cdot)$, the mutual information between $\mathbf{x}_\URUT$ and $\mathbf{y}_\UR$ is defined as
\begin{align}\label{eq:ixyr}
    I(\mathbf{x}_\URUT;\mathbf{y}_\UR,\bm{\Theta}_\UR)=h(\mathbf{x}_\URUT|\bm{\Theta}_\UR)-h(\mathbf{x}_\URUT|\mathbf{y}_\UR,\bm{\Theta}_\UR).
\end{align}
Following~\cite[Appendix C]{9079911}, under Gaussian signaling, we obtain $h(\mathbf{x}_\URUT|\bm{\Theta}_\UR)=\log_2(\det(\pi e \mathbf{I}_\NR))$. Next, following~\cite[Appendix I]{1624653}, $h(\mathbf{x}_\URUT|\mathbf{y}_\UR,\bm{\Theta}_\UR)$ is upper bounded by
\begin{align}
    h(\mathbf{x}_\URUT|\mathbf{y}_\UR,\bm{\Theta}_\UR)&\leq\esp\Bigl\{\log_2\left(\det\left(\pi e \esp\bigl\{\epsilon_\URUT \epsilon_\URUT^H|\bm{\Theta}_\UR\bigr\}\right)\right)\Bigr\},\label{eq:hxy}
\end{align}
where $\epsilon_\URUT$ is the MMSE estimation error of $\mathbf{x}_\URUT$ given $\mathbf{y}_\UR$ and $\bm{\Theta}_\UR$. Accordingly, $\esp\{\epsilon_\URUT \epsilon_\URUT^H|\bm{\Theta}_\UR\}$ is computed as 
\begin{align}~\label{eq:xuruthat}
    \esp\{\epsilon_\URUT \epsilon_\URUT^H|\bm{\Theta}_\UR\}&=\mathbf{C}_{\mathbf{x}_\URUT\mathbf{x}_\URUT|\bm{\Theta}_\UR}-\mathbf{C}_{\mathbf{x}_\URUT\mathbf{y}_\UR|\bm{\Theta}_\UR}\mathbf{C}_{\mathbf{y}_\UR\mathbf{y}_\UR|\bm{\Theta}_\UR}^{-1}\mathbf{C}_{\mathbf{y}_\UR\mathbf{x}_\URUT|\bm{\Theta}_\UR}.
\end{align}

The covariance matrices in~\eqref{eq:xuruthat} are calculated as
\begin{align}
\mathbf{C}_{\mathbf{x}_\URUT\mathbf{x}_\URUT|\bm{\Theta}_\UR}&=\esp\Bigl\{\mathbf{x}_\URUT\mathbf{x}_\URUT^H|\bm{\Theta}_\UR\Bigr\}\!=\!\mathbf{I}_\NR,\label{eq:cxurxur}\\
    \mathbf{C}_{\mathbf{x}_\URUT\mathbf{y}_\UR|\bm{\Theta}_\UR}&=\esp\Bigl\{\mathbf{x}_\URUT\mathbf{y}_\UR^H|\bm{\Theta}_\UR\Bigr\}\!=\!\sqrt{\rho_\UT}\esp\Bigl\{\bm{\Lambda}_\UR^{1/2}\!\mathbf{A}_\UR^H|\bm{\Theta}_\UR\Bigr\},\label{eq:cxuryr}\\
    \mathbf{C}_{\mathbf{y}_\UR\mathbf{y}_\UR|\bm{\Theta}_\UR}&=\esp\Bigl\{\mathbf{y}_\UR\mathbf{y}_\UR^H|\bm{\Theta}_\UR\Bigr\}\!=\!\mathbf{I}_{N_\UR}\!\!+\!\rho_{\mathrm{J}}\esp\Bigl\{\Frj(\Frj)^H\!\Bigr\}\nonumber\\\!&+\!\rho_\UT\esp\Bigl\{\mathbf{A}_\UR\bm{\Lambda}_\UR\mathbf{A}_\UR^H|\bm{\Theta}_\UR\Bigr\},\label{eq:cyryr}\\
   \mathbf{C}_{\mathbf{y}_\UR\mathbf{x}_\URUT|\bm{\Theta}_\UR}&=\mathbf{C}_{\mathbf{x}_\URUT\mathbf{y}_\UR|\bm{\Theta}_\UR}^H=\sqrt{\rho_\UT}\esp\Bigl\{\mathbf{A}_\UR\bm{\Lambda}_\UR^{1/2}|\bm{\Theta}_\UR\Bigr\}\label{eq:yrxur}.
\end{align}
By plugging (\ref{eq:cxurxur})--(\ref{eq:yrxur}) into \eqref{eq:xuruthat}, and then replacing $h(\mathbf{x}_\URUT|\bm{\Theta}_\UR)$ and \eqref{eq:hxy} into \eqref{eq:ixyr}, the SE at UR can be computed as in \eqref{eq:SEr} by employing the matrix inversion lemma. Analogously, to obtain the SE at the CPU, the mutual information in \eqref{eq:tildexrt} must be computed between $\mathbf{x}_\URUT$ and ${\mathbf{z}}_\MN$. 
\bibliographystyle{IEEEtran}   
\bibliography{references}

\end{document}